\documentclass[journal, comsoc]{IEEEtran}

\IEEEoverridecommandlockouts
\usepackage{cite}
\usepackage{amsmath,amssymb,amsfonts}

\usepackage{algorithmic}
\usepackage{graphicx, subfigure}
\usepackage{textcomp}
\usepackage{xcolor}
\usepackage{amsthm}
\usepackage{verbatim}
\usepackage{multirow}
\usepackage{bbm}
\usepackage{float}

\newtheorem{theorem}{Theorem}
\newtheorem{proposition}{Proposition}
\newtheorem{lemma}{Lemma}
\newtheorem{remark}{Remark}
\newtheorem{corollary}{Corollary}

\def\bu{{\bf u}}

\def\bx{{\bf x}}

\def\cA{\mbox{$\mathcal{A}$}}

\def\bbE{\mbox{$\mathbb{E}$}}

\def\bbP{\mbox{$\mathbb{P}$}}

\def\d{{\rm d}}

\makeatletter
\def\blfootnote{\xdef\@thefnmark{}\@footnotetext}
\makeatother

\usepackage{suffix}
\usepackage{mathtools}

\DeclarePairedDelimiterX\MeijerM[3]{\lparen}{\rparen}%
{\begin{smallmatrix}#1 \\ #2\end{smallmatrix}\delimsize\vert\,#3}

\newcommand\MeijerG[8][]{%
  G^{\,#2,#3}_{#4,#5}\MeijerM[#1]{#6}{#7}{#8}}

\WithSuffix\newcommand\MeijerG*[7]{%
  G^{\,#1,#2}_{#3,#4}\MeijerM*{#5}{#6}{#7}}

\def\BibTeX{{\rm B\kern-.05em{\sc i\kern-.025em b}\ker n-.08em
    T\kern-.1667em\lower.7ex\hbox{E}\kern-.125emX}}

\begin{document}

\title{ 
Low-Earth Orbit Satellite Network Analysis: Coverage under Distance-Dependent Shadowing
}

\author{
\IEEEauthorblockN{Jinseok~Choi},
{\it Member,~IEEE},
\and
\IEEEauthorblockN{ Jeonghun~Park},
{\it  Member,~IEEE},
\IEEEauthorblockN{ Junse~Lee},
{\it  Member,~IEEE},
\\and 
\IEEEauthorblockN{Namyoon~Lee},
 {\it Senior Member,~IEEE}

\thanks{ J. Choi is with School of Electrical Engineering, Korea Advanced National Institute of Science and Technology, Republic of Korea (e-mail: {\texttt{jinseok@kaist.ac.kr}}).

 J. Park is with School of Electrical and Electronic Engineering, Yonsei University, Seoul, Republic of Korea (e-mail: {\texttt{jhpark@yonsei.ac.kr}}). 
 
 J. Lee is with School of AI Convergence, Sungshin Women’s University, Seoul,  Republic of Korea (e-mail: {\texttt{junselee@sungshin.ac.kr}}).
 
 N. Lee is with Department of Electrical Engineering, Korea University, Seoul, Republic of Korea (e-mail: {\texttt{namyoon@korea.ac.kr}})}
}

\maketitle

\begin{abstract} 
This paper offers a thorough analysis of the coverage performance of Low Earth Orbit (LEO) satellite networks using a strongest satellite association approach, with a particular emphasis on shadowing effects modeled through a Poisson point process (PPP)-based network framework. We derive an analytical expression for the coverage probability, which incorporates key system parameters and a distance-dependent shadowing probability function, explicitly accounting for both line-of-sight and non-line-of-sight propagation channels. To enhance the practical relevance of our findings, we provide both lower and upper bounds for the coverage probability and introduce a closed-form solution based on a simplified shadowing model. Our analysis reveals several important network design insights, including the enhancement of coverage probability by distance-dependent shadowing effects and the identification of an optimal satellite altitude that balances beam gain benefits with interference drawbacks. Notably, our PPP-based network model shows strong alignment with other established models, confirming its accuracy and applicability across a variety of satellite network configurations. The insights gained from our analysis are valuable for optimizing LEO satellite deployment strategies and improving network performance in diverse scenarios.
\end{abstract}


\blfootnote{This work was presented in part at the {\em IEEE International Conference on Communications (IEEE ICC)}, Rome, Italy, 2023 \cite{choi2023coverage}.}



\section{Introduction}

As the communications industry moves toward the era of 6G, it is evident that satellite networks will be essential in achieving the vision of ubiquitous global connectivity. Integrating satellite communication systems into the 6G ecosystem holds the promise of extending coverage to remote and underserved areas, enhancing network resilience, and supporting a wide array of wireless applications, ranging from the Internet of Things (IoT) to autonomous vehicles \cite{xiao2024wcmag}. Among these, low Earth orbit (LEO) satellite communications have attracted considerable attention due to their advantageous features, such as relatively low propagation delay and the ability to support dense deployments. This growing interest emphasizes the need for a comprehensive understanding of the communication performance achievable by dense LEO satellite networks \cite{wang2022commmag}, as precise performance evaluation is crucial for the effective design, deployment, and operation of these networks.


\subsection{Prior Works}
Recently, there has been a significant increase in research focused on analyzing the coverage performance of LEO satellite networks, with stochastic geometry emerging as a popular approach \cite{park2022tractable, okati2020downlink, okati2022nonhomo, wang2022commmag}. Stochastic geometry, a mathematical framework for modeling the spatial distribution of wireless nodes, initially proved highly effective in terrestrial cellular networks. Notably, a Poisson point process (PPP) was used in \cite{andrews:tcom:11} to model base station (BS) locations, leading to a tractable expression for coverage probability. This framework has since been extended to include heterogeneous networks \cite{dhillon2012modeling, jo12hetnet, park2018hetnet,lee2014power}, multi-cell coordination \cite{park16coloring, lee2014spectral}, MIMO systems \cite{park2016optimal, bai2016uplnkmimo, choi2017letter}, and sensing-integrated systems \cite{park2018sensing, olson2024sensing}.


A key research direction in this field involves analyzing terrestrial cellular networks while accounting for shadowing effects. High-frequency signals, such as those in the millimeter wave bands, are particularly susceptible to obstructions, leading to significant attenuation and penetration loss due to blockages. This results in distinct propagation characteristics, such as differing path-loss exponents or path gains, between line-of-sight (LOS) and non-line-of-sight (NLOS) environments \cite{pi2011commmag}. Notably, this disparity between LOS and NLOS conditions is observed not only in the millimeter wave range (24 GHz to 52 GHz) but also in lower frequency bands (300 MHz to 3 GHz) and higher frequency bands \cite{3gpp2010uhf}. To address this, \cite{bai2014coverage, bai2014coverageTWC} introduced a novel analytical technique that randomly marks each communication link from a BS as either LOS or NLOS. In \cite{bai2014coverage}, the LOS probability is modeled as an exponential function of the link distance. Leveraging this, a millimeter wave cellular network was properly modeled by using a marked Poisson point process \cite{baccelli:book:09}, by which the coverage probability was characterized by incorporating LOS/NLOS distinctions. 
Later, this approach was extended by incorporating millimeter wave MIMO \cite{kulkarni2016tcom} and multi-cell coordination under shadowing effects \cite{park2018inter}.








Building on the success of stochastic geometry in analyzing terrestrial cellular networks, significant progress has been made in the coverage performance analysis of satellite networks. In \cite{okati2020downlink}, a binomial point process (BPP) was used to model the spatial locations of LEO satellites, leading to the derivation of a coverage probability expression. This BPP-based framework was later extended to incorporate shadowed-Rician fading \cite{jung2022sr} and non-homogeneous satellite deployments \cite{okati2022nonhomo}. In \cite{talgat2020nearest, talgat2020stochastic}, the distribution of link distances between a user and the nearest satellite was characterized for multiple orbital satellite networks, which enabled the calculation of coverage probability.

Complementing these BPP models, our previous work \cite{park2022tractable} introduced a novel spherical PPP model to represent LEO satellite constellations. Beyond theoretical developments, \cite{park2022tractable} validated this model by comparing analytical coverage probabilities with data from actual Starlink constellations. Leveraging this approach, \cite{kim2024coord} analytically investigated the benefits of multi-satellite coordination, and \cite{park:arxiv:23} developed a unified modeling framework that integrates LEO satellite networks with terrestrial cellular networks. Furthermore, \cite{chae2023performance} incorporated beamwidth considerations and LOS/NLOS distinctions into their analysis. In addition to these studies, several other works have advanced satellite network analysis \cite{al2021optimal, al2021analytic, wang_alouni2024jsac, kim2023arxiv}.

Despite significant progress in understanding the coverage performance of LEO satellite networks, an essential aspect of channel modeling—shadowing—has not been sufficiently addressed in existing coverage analysis studies. Thoroughly incorporating shadowing into the coverage analysis of LEO satellite networks is crucial for several reasons. 
First, satellite communications primarily operate in higher frequency bands, such as the Ku band (12–18 GHz) and Ka band (26.5–40 GHz). These bands are particularly vulnerable to shadowing, making signals more prone to blockage and attenuation by obstacles. Second, in remote or underserved areas, where satellite communication may be the only available wireless connection, any disruption caused by shadowing can lead to a complete loss of service. Third, the long propagation distances between satellites and ground users make these communication links highly dependent on an unobstructed LOS. Blockages caused by buildings, mountains, or other obstacles can result in significant signal degradation or complete loss of communication. Unlike terrestrial networks, where signals can often find alternative paths through reflection or diffraction, satellite signals have limited alternative routes due to the direct, long-distance nature of the connection.


Some prior works have considered shadowing effects in satellite coverage analysis, but their approaches are limited. For example, studies like \cite{jung2022sr, talgat2020stochastic, okati2022nonhomo} used tailored fading models, such as the shadowed-Rician model, to capture channel attenuation due to shadowing. However, these methods do not account for LOS/NLOS distinctions for individual links, as the same fading is applied uniformly. Other studies, such as \cite{al2021analytic, chae2023performance}, assumed LOS or NLOS conditions for links without thoroughly modeling the LOS/NLOS probability. For instance, \cite{chae2023performance} used an approximate LOS ball approach from \cite{bai2014coverageTWC}, where links beyond a certain distance threshold were classified as NLOS. A more precise approach, as demonstrated in \cite{al2020modeling}, would characterize LOS probability based on link distance or elevation angle. This paper aims to advance the comprehensive integration of shadowing effects into LEO satellite analysis, distinguishing it from previous studies.

\subsection{Contributions}
This paper presents an analysis of LEO satellite network coverage performance with Earth-moving beams, emphasizing the impact of shadowing effects using a PPP-based network model. The primary contributions are summarized as follows:

\begin{itemize}
    \item We derive an analytical expression for coverage probability in the interference-limited regime, incorporating key system parameters and a distance-dependent shadowing probability function. Our model explicitly distinguishes between LOS and NLOS propagation channels, using distinct fading and path loss exponents for a more accurate representation of real-world conditions compared to previous studies \cite{al2021analytic, park2022tractable,okati2022nonhomo}. Additionally, we implement a strongest satellite association rule, recognizing that nearest association is invalid due to shadowing, offering a more realistic portrayal of satellite network dynamics and enabling more accurate performance predictions. This expression provides a robust foundation for optimizing LEO satellite network design.
    
    \item To enhance the practical applicability of our findings, we obtain lower and upper bounds for the coverage expression by addressing the integral operator concerning the distance to the strongest satellite. Using a simplified shadowing model, we present a closed-form coverage expression, making the results more accessible. This expression reveals that a significant increase in satellite density can reduce coverage probability. We further identify the optimal density under a strongest satellite association rule in a non-shadowing scenario.
    
    \item Our numerical study highlights several significant findings. Notably, distance-dependent shadowing effects enhance coverage probability by reducing interference. We also observe that denser urban networks exhibit less sensitivity to changes in satellite density due to the high shadowing probability. Additionally, the optimal satellite altitude, determined to be between 500 and 700 km for the considered system, balances the benefits of increased beam gain against the drawbacks of wider beam interference.
    
    \item Finally, our PPP-based network model demonstrates remarkable consistency with other established models. The coverage probability results derived from our approach align closely with those from BPP, Walker star constellation, and Starlink network models. This concordance validates the accuracy of our analytical framework and underscores its versatility across various satellite network configurations. By providing a comprehensive yet flexible model that accounts for channel shadowing effects, our research offers valuable insights for optimizing LEO satellite deployment strategies and enhancing network performance across a wide range of scenarios.
\end{itemize}

\section{System Model}
\label{sec:system}
In this section, we describe the considered downlink LEO satellite communication network model and channel model.

\subsection{Network Model}

We model the Earth as a sphere with radius $R_{\sf E}$. 
The locations of users on Earth's surface are represented by independent and homogeneous Poisson point processes (PPPs).
The satellites orbit  Earth at an altitude $h$, forming a sphere with radius $R_{\sf S} = R_{\sf E} + h$ as shown in Fig.~\ref{fig:network}. 
The locations of satellites on this sphere are also modeled as independent and homogeneous PPPs.

To be specific, we denote the surface of the satellite sphere in $\mathbb{R}^3$ with the center at the origin ${\bf 0}\in \mathbb{R}^3$ and  radius $R_{\sf S}$ as
\begin{align}
	\mathbb{S}_{R_{\sf S}}^2=\{{\bf x}\in \mathbb{R}^3: \|{\bf x}\|_2=R_{\sf S}\}.
\end{align}  
Any point vector on this sphere ${\bf x}\in \mathbb{S}_{R_{\sf S}}^2$ can be represented using a polar coordinate system, with a pair of angles: azimuth angle $0\leq \theta \leq 2\pi$ and elevation angle $0\leq \phi \leq 2\pi$. 
Let $\Phi=\{{\bf x}_1,\ldots, {\bf x}_N\}$ be  a homogeneous spherical PPP (SPPP) with a finite number of elements on the surface of the sphere $\mathbb{S}_{R_{\sf S}}^2$.
In addition $\Phi(\mathbb{S}_{R_{\sf S}}^2)=N$ denotes the number of points on  $\mathbb{S}_{R_{\sf S}}^2$, and the variable $N$ follows a Poisson distribution with mean of $4\lambda\pi R_{\sf S}^2 $ where $\lambda$ is the density of the homogeneous SPPP.
We assume that satellites are distributed according to this homogeneous SPPP  with density $\lambda$, i.e., $\Phi=\{{\bf x}_1,\ldots, {\bf x}_{N}\}$.
The probability density function (PDF) of the number of satellites $N$ is given by
\begin{align}
	\mathbb{P}\left( N=n \right) =e^{-4\pi R_{\sf S}^2\lambda}\frac{\left(4\pi R_{\sf S}^2\lambda\right)^n}{n!}, \label{eq:poisson}
\end{align} where $|\mathbb{S}_{R_{\sf S}}^2|=4\pi R_{\sf S}^2$ is the surface area of the sphere. 
We note that for given $\Phi(\mathbb{S}_{R_{\sf S}}^2)=N$, the point $\{{\bf x}_1,\ldots, {\bf x}_N\}$ follows a BPP.
In this process, each point ${\bf x}_i$  is independent and uniformly distributed on the surface of the satellite sphere.
Throughout the paper, we use $e^{x}$ and $\exp{(x)}$ interchangeably for notational simplicity.

We now introduce $\Phi_{\sf U} = \{ {\bf u}_{1},  \ldots, {\bf u}_{M}\}$, which is a homogeneous SPPP representing the distribution of users on the surface of  Earth, denoted as  $\mathbb{S}_{R_{\sf E}}^2$; the users are distributed on $\mathbb{S}_{R_{\sf E}}^2$ according to the homogeneous SPPP $\Phi_{\sf U}$ with density $\lambda_{\sf U}$. 
The number of users, $M$, follows a Poisson distribution with a mean value of $4\lambda_{\sf U}\pi R_{\sf E}^2 $.
It is important to note that the user distribution process $\Phi_{\sf U}$ is independent of the underlying satellite placement process $\Phi$.


Using Slivnyak's theorem~ \cite{baccelli2009stochasticGeo}, we consider a typical user to be located at $\bu_1 = (0,0,R_{\sf E})$ on  $\mathbb{S}_{R_{\sf E}}^2$, without loss of generality. 
Throughout this paper, we use  $\bu_1$ to refer to the location of the typical user.
From the perspective of the typical user, we define a corresponding typical spherical cap $\mathcal{A}\subset  \mathbb{S}_{R_{\sf S}}^2$, which is a subset of the satellite sphere surface  $\mathbb{S}_{R_{\sf S}}^2$ as shown in  Fig.~\ref{fig:network}.
This typical spherical cap is the partial surface of the sphere  $\mathbb{S}_{R_{\sf S}}^2$ that is cut off by a tangent plane to the Earth's surface  $\mathbb{S}_{R_{\sf E}}^2$ at the typical user's location $\bu_1 = (0,0,R_{\sf E})$.
The area of the typical spherical cap $\cA$ is \cite{park2022tractable}
\begin{align}
    \label{eq:area}
	|\mathcal{A}|=2\pi (R_{\sf S}-R_{\sf E})R_{\sf S}.
\end{align}
We assume that satellites located on the spherical cap $\cA$ only are capable of communicating with the typical user, and called visible satellites.
We also define the average number of visible satellites as $K = \lambda |\mathcal{A}|$.

\begin{figure} 
    \centering 
    \includegraphics[width=0.95\columnwidth]{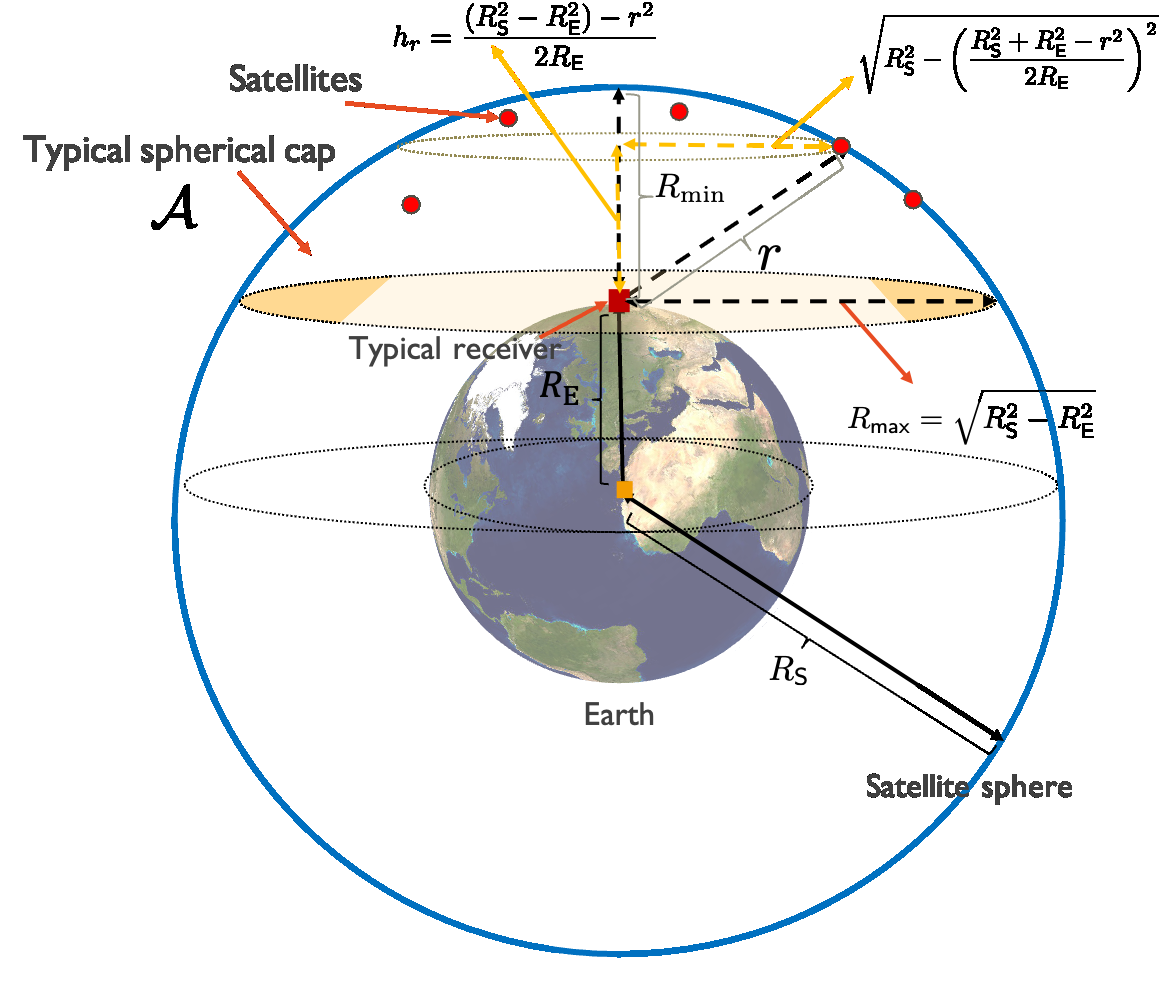}
    \caption{Satellites are assumed to be distributed on the surface of the satellite sphere with radius of $R_{\sf S} = R_{\sf E} + h $ where  $h$ is the satellite altitude  and $R_{\sf E}$ is  the radius of Earth. 
    A typical user is located at $(0,0,R_{\sf E})$ and can only be served by satellites on the typical spherical cap $\cA$.}
    \label{fig:network}
\end{figure}


\subsection{Pathloss and Fading Models }

In wireless communication systems, the channels are affected by both large-scale fading and small-scale fading. 
To model the large-scale fading, we employ the classical pathloss model, which depends on the distance between satellite $i\in[N]$ and the typical user, as well as the corresponding pathloss exponent $\alpha_i$.
The pathloss of satellite $i$ to the typical user can be expressed as:
\begin{align}
	  \|{\bf x}_{i}-{\bf u}_1\|^{-\alpha_i} = r_i^{-\alpha_i}.
\end{align}
In this model, we consider two types of pathloss exponents to incorporate channel shadowing effect as shown in Fig.~\ref{fig:losnlos}: $\alpha_i \in \{\alpha_{\sf L},\alpha_{\sf N} \}$. 
The pathloss exponent $\alpha_{\sf L}$ corresponds to the case where satellite $i$ experiences LOS propagation to the typical user, while $\alpha_{\sf L}$ represents the case of NLOS propagation.
\begin{figure} 
    \centering 
    \includegraphics[width=0.95\columnwidth]{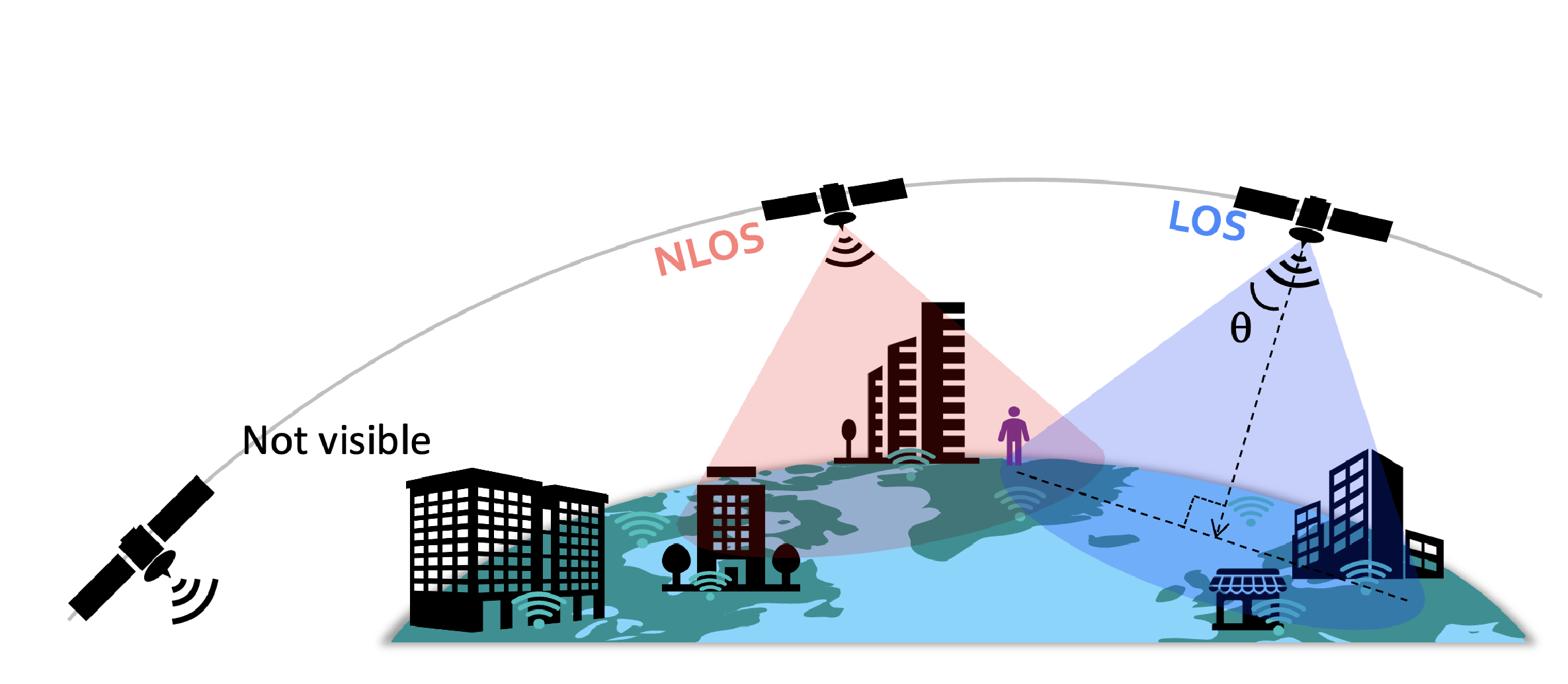}
    \caption{Satellites have different LOS and NLOS channel probability depending on the distance to the typical user. In addition, they also have different beam gain depending on the distance to the typical user. }
    \label{fig:losnlos}
\end{figure}

To capture the randomness of shadowing effect, we adopt a distance-dependent LOS probability model. 
Let  $p_{\sf L}(r)$ be the probability that a satellite located at a distance $r$ from the typical user experiences LOS propagation.
Let $H_i$ denote the small-scale fading gain from satellite $i$ to the typical user. 
Then the distribution of the small-scale channel fading under the distance-dependent shadowing channel can be modeled to have the following PDF:
\begin{align}
    \label{eq:LosNlos}
    f_{\sqrt{H_i}}(x) =\; &p_{\sf L}(r_i)f_{\sqrt{H_i}|\sf L}(x) + (1- p_{\sf L}(r_i))f_{\sqrt{H_i}|\sf N}(x),
\end{align}
where $f_{\sqrt{H_i}|\sf L}(x)$ and $f_{\sqrt{H_i}|\sf N}(x)$ indicate the PDFs of $\sqrt{H_i}$ for the LOS and NLOS channels, respectively. 
For the LOS channel $f_{\sqrt{H_i}|\sf L}(x)$, we assume the Nakagami-$m$  distribution to suitably capture the LOS fading effect.
Assuming $\mathbb{E}[H_i]=1$, $f_{\sqrt{H_i}|\sf L}(x)$ is given by \cite{giunta2018estimation}:
\begin{align}
    \label{eq:LosFading}
	 f_{\sqrt{H_i}|\sf L}	(x)&= \frac{2m^m}{\Gamma(m)}x^{2m-1}\exp\left(-mx^2\right),
\end{align}
for $x\geq 0$.  
For the NLOS fading channel, we consider a Rayleigh distribution. 
Since the Nakagami-$m$ distribution reduces to the Rayleigh distribution when $m=1$, $f_{\sqrt{H_i}|\sf N}(x)$ can be directly obtained from \eqref{eq:LosFading} by setting $m=1$.

Now, we introduce the transmit beam gains involved in the communication between satellites and the typical user, assuming the receive beam gain is one, i.e., omni-directional antenna at the typical user.
An example of the considered scenario can be {\it Starlink Direct to Cell} in which  satellites act as cell towers in space.
In addition, we assume the scenario of fixed beamforming (Earth-moving beamforming) from satellites where satellites 
Then it is appropriate to consider the beamforming gain as a function of the distance between the typical user and the $i$th satellite.
Accordingly, we denote the effective beamforming gain from satellite $i$ to the typical user as $G(r_i)$.

\begin{remark} [Nakagami-$m$ and shadowed-Rician fading]\normalfont

As demonstrated in \cite{abdi2003simple}, the shadowed-Rician fading model is widely used to accurately represent satellite channel characteristics. However, recent findings in \cite{talgat2024leouplink} have shown that the shadowed-Rician fading distribution can be closely approximated by the Gamma distribution through moment matching. Notably, in heavy and average shadowing scenarios, the Gamma distribution provides an almost exact representation of the shadowed-Rician fading distribution. Since the Nakagami-$m$ fading distribution is also a Gamma distribution, our analytical results can be easily extended to scenarios involving shadowed-Rician fading. Thus, our derivation inherently encompasses shadowed-Rician fading.
\end{remark}


\section{Coverage Probability Analysis}
\label{sec:coverage}

In this section, we begin by introducing the coverage probability as a key performance metric for analyzing downlink satellite networks within the considered network and channel model. Following this, we provide an analytical derivation of the coverage probability across various scenarios to evaluate the network's performance.

\subsection{Performance Metric}    

We assume that the typical user is served by the satellite which provides the strongest channel gain to the user.
This strongest-satellite association rule is different from the nearest-satellite association rule in~\cite{park2022tractable} in which channels are experience the homogeneous fading, i.e., same pathloss exponent and fading distribution for all satellites.
The strongest-satellite association rule is valid for the considered network since the channel quality cannot be determined solely by the distance to the typical user due to the random shadowing effect.

Let ${\bf x}_{i}\in \Phi$ be the location of the associated satellite, and $P$ be the transmit power.
Then the signal-to-interference-plus-noise (SINR) of the typical use is
\begin{align}
  	{\sf SINR}=\frac{P G(r_i)H_i\|{\bf x}_{i}-{\bf u}_1\|^{-\alpha_i}}{\sum_{{\bf x}_j \in \Phi_{{\rm I}({\bf x}_{i})}} P G(r_j)H_j\|{\bf x}_{j}-{\bf u}_1\|^{-\alpha_j}+\sigma^2},
\end{align} 
where $\sigma^2$ denotes the noise power and  $\Phi_{{\rm I}({{\bf x}_i})}$ represents the set of  satellites that cause interference to the typical user when the associated satellite is located at ${\bf x}_i$.
Using the SINR expression, the coverage probability is given as
\begin{align}
    \nonumber
    P^{\sf  cov}_{{\sf SINR}} (\gamma; \lambda, R_{\sf S}) &=\mathbb{P}\left[{\sf SINR}\geq \gamma  \right] 
    \\\label{eq:Pcov_def}
    &= \! \mathbb{P}\left[\!\frac{G(r_i)H_i\|{\bf x}_{i}-{\bf u}_1\|^{-\alpha_i}}{\sum_{{\bf x}_j \in \Phi_{{\rm I}({\bf x}_{i})}}G(r_j)H_j\|{\bf x}_{j}\!-\!{\bf u}_1\|^{-\alpha_j}\!+\!{\bar \sigma}^2}\geq \gamma \right]\!,
\end{align}
where ${\bar\sigma}^{2}=\frac{\sigma^2}{P}$.
The coverage probability in \eqref{eq:Pcov_def} considers several important factors that influence the performance of downlink satellite networks, such as satellite availability for the typical user, pathloss, shadowing, fading distribution, satellite distribution density, and satellite altitude. 
By deriving a closed-form expression for the coverage probability that eliminates any random variables, we can gain valuable insights into the overall behavior and performance of the system at a network level. 

\subsection{Coverage Probability}

To derive a deterministic formula for the coverage probability given in \eqref{eq:Pcov_def}, we begin by finding the Laplace transform of the total interference. 
The Laplace transform plays a key role in solving the expression for coverage probability. 
We define the total interference power as follows:
\begin{align}
    \label{eq:interference}
    I_{r_i}=\sum_{{\bf x}_j\in \Phi_{{\rm I}({{\bf x}_i})} } G(r_j)H_j\|{\bf x}_j-{\bf u}_1\|^{-\alpha_j}.
\end{align} 
Then the Laplace transform of $I_{r_i}$ in \eqref{eq:interference} is derived in the following lemma:
\begin{lemma} [Interference Laplace]
    \label{lem:laplace}
    Conditioned on that the distance between the typical receiver and the associated satellite is $r$, the Laplace transform of the aggregated interference power is derived as 
        \begin{align}
           \nonumber
            \mathcal{L}_{I_{r}}(s) = &\exp\Bigg(-2\pi \lambda \frac{R_{\sf S}}{R_{\sf E}}\int_{R_{\rm min}}^{R_{\rm max} } \Bigg( 1-  p_{\sf L}(v)\frac{1}{ \left(1+\frac{sv^{-\alpha_{\sf L}}G(v)}{m}\right)^{m}} 
            \\ \label{eq:Laplace_derived}
            &- (1-p_{\sf L}(v))\frac{1}{1+sv^{-\alpha_{\sf N}}G(v)}\Bigg) v~ {\rm d}v \Bigg).
        \end{align}
\end{lemma}
\proof 
    See Appendix~\ref{app:laplace}.
\endproof

Based on Lemma~\ref{lem:laplace}, we perform our analysis by deriving an exact deterministic expression for the coverage probability in \eqref{eq:Pcov_def}.
In particular, we focus on the interference-limited regime in which the network-level analysis is more meaningful \cite{park2022tractable}. 
The coverage probability is derived in the following theorem, which is the main technical result of this paper:
\begin{theorem}[Coverage probability]\label{thm:Pcov_exact} 
        In the interference-limited regime, i.e., $I_{r}\gg {\bar\sigma}^{2}$, the coverage probability of the typical receiver for target signal-to-interference ratio (SIR) $\gamma > 0$ dB is derived as
    \begin{align}
        \nonumber
        &P^{ {\sf cov}}_{{\sf SIR}} (\gamma; \lambda, R_{\sf S}) 
        \\\nonumber
        &= 2\pi\lambda\frac{R_{\sf S}}{R_{\sf E}}\int_{R_{\rm min}}^{R_{\rm max}}  \Bigg( p_{\sf L}(r)\sum_{k=0}^{ m -1}\frac{ (-m)^k\gamma^k{r}^{k\alpha_{\sf L} }}{k!G^k(r) }\left.{\frac{\d^k\mathcal{L}_{{  I}_{r} }(s)}{\d s^k}} \right|_{s=\frac{ m\gamma  {r}^{\alpha_{\sf L}}}{G(r)}} 
        \\ \label{eq:CovProb_derived}
        &\quad +  \big(1-p_{\sf L}(r)\big)\mathcal{L}_{{  I}_{r} }\left(\frac{\gamma r^{\alpha_{\sf N}}}{G(r)}\right)\Bigg) r {\rm d} r.
    \end{align}
\end{theorem} 

\begin{proof}
    See Appendix~\ref{app:Pcov_exact}.
\end{proof}

It is important to note that the coverage probability expression derived in \eqref{eq:CovProb_derived} is exact for target SIR $\gamma$ greater than $0$ dB. 
We will validate the derived coverage probability and also demonstrate that although the derived expression is exact for 
$\gamma > 0$ dB, it still closely follows the general trend of the numerical coverage probability for $\gamma < 0$ dB as an upper bound, providing a reasonable approximation in that range.
This allows us to gain insights into the performance of satellite networks not only in the medium-to-high SIR regime but also in the low-to-medium SIR regime by using the derived expression.

For the validation, we use the exponential blockage probability distribution \cite{al2020modeling} and Bessel beam gain model \cite{diaz2007non} to model the LOS probability $p_{\sf L}(r)$ and beam gain $G(r)$, respectively.
We remark that these models incorporate the dependency on the satellite altitude and elevation angles from the receiver tangential plan.
The LOS probability is modeled as \cite{al2020modeling}:
\begin{align}
    \label{eq:Plos_exp}
    p_{\sf L}(r) \!= \!\exp\left(-\beta \cot\left(\arcsin\left(\frac{(R_{\sf E} \!+ \!h)^2-R_{\sf E}^2}{2rR_{\sf E}} \!- \!\frac{r}{2R_{\sf E}}\right)\right)\right),
\end{align}
where $\beta \geq 0$ is a constant related to the geometry of the urban environment.
This model in \eqref{eq:Plos_exp} presents a reasonable agreement with the empirical LOS probabilities in 3GPP model \cite{al2020modeling,3gpp2018study} by adjusting the parameter  $\beta$.
As  $\beta$ increases, the probability of having a LOS channel  decreases, and vice versa.
The beam gain is modeled as \cite{diaz2007non}
\begin{align}
    \label{eq:bessel}
    G(r) = G_{\rm max} \left(\frac{J_1(u)}{2u} + 36\frac{J_3(u)}{u^3}\right)^2,
\end{align}
where $u = 2.07123 \sin\theta/\sin\theta_{\rm 3 dB}$, $\theta$ is the angle between the beam center and  the typical user as shown in Fig.~\ref{fig:losnlos}, $\theta_{\rm 3 dB}$ is half-power beamwidth, and $G_{\rm max}$ is the maximum antenna gain.
For the considered network, we have \[\sin\theta = \sqrt{-\frac{(h^2 - r^2)(h^2 - r^2 + 4hR_{\sf E} + 4R_{\sf E}^2)}{4r^2(h + R_{\sf E})^2 }}.\]

\begin{figure} 
    \centering 
    \includegraphics[width=0.95\columnwidth]{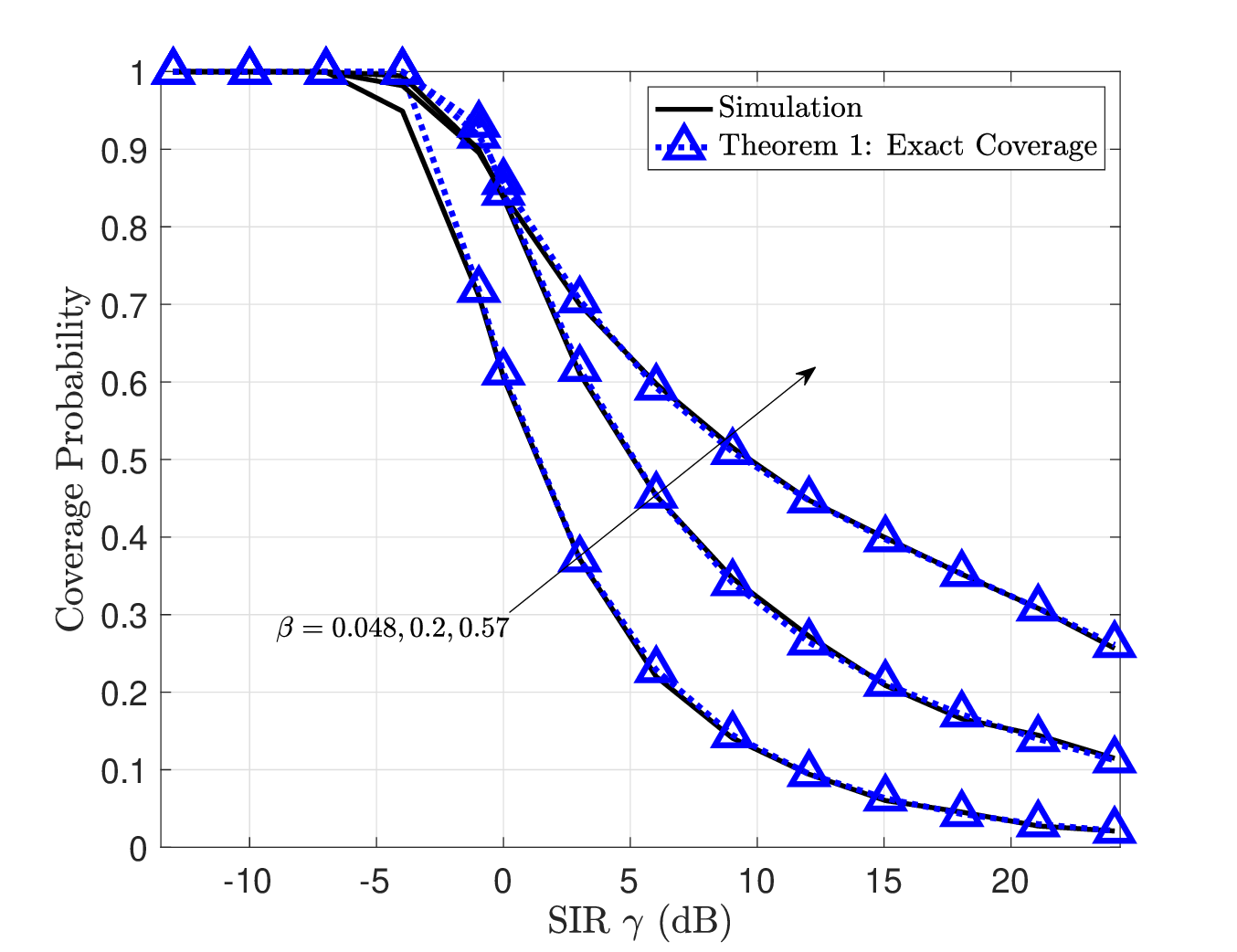}
    \caption{The numerical and analytical coverage probabilities versus the target SIR $\gamma$ for $\alpha_{\sf L} = 2$, $\alpha_{\sf N} =3$, $h = 700$ km,  $m=3$,  $K=10$, and $\beta \in \{0.048, 0.2, 0.57\}$.}
    \label{fig:Pcov_exact}
\end{figure}

In Fig.~\ref{fig:Pcov_exact}, we consider  $G_{\rm max} = 20$ dB and $\theta_{\rm 3 dB} = 0.2^\circ$ with  $\alpha_{\sf L} = 2$, $\alpha_{\sf N} =3$, $h = 700$ km,  $m=3$,  $K=10$, and $\beta \in \{0.048, 0.2, 0.57\}$ \cite{al2020modeling}.
Throughout the paper, we plot $\min(P_{\sf SIR}^{\sf cov},1)$  for the analytical expression as it is exact for $\gamma > 0$ dB and considered to be an upper bound for $\gamma \leq 0$ dB, unless mentioned otherwise.
As noted in Fig.~\ref{fig:Pcov_exact}, the derived coverage expression exactly matches with the numerical results for $\gamma > 0$ dB. 
For $\gamma \leq 0 $ dB, the analytical results also reveals reasonable accuracy with almost the same curvature as an upper bound. 
A key observation is that as the environment becomes denser (larger $\beta$), the coverage probability increases as the interfering satellites more likely to have NLOS propagation paths, which decreases interference.
Although the derived coverage expression in Theorem~\ref{thm:Pcov_exact} is exact for $\gamma > 0$ dB and useful for $\gamma \leq 0 $ dB and thus, we can obtain meaningful network insight by evaluating the expression, it is still desirable to further refine the expression for better understanding of the network behaviour according to the relevant network parameters.

The key difficulty in refining the derived analytical expression mainly comes from the derivatives of the interference Laplace and the integral in the interference Laplace. 
To resolve this challenge, we derive bounds of the coverage probability in \eqref{eq:Pcov_def} as a stepping stone for further analysis.
\begin{theorem}\label{thm:SandwichBounds}
    In the interference-limited regime, the coverage probability  in \eqref{eq:Pcov_def} is upper and lower bounded for $\gamma > 0 $ dB  as
    \begin{align}
        \nonumber
        P_{\sf SIR}^{\sf cov, b}(\gamma; \lambda, R_{\sf S}, 1) \!\leq\! P_{\sf SIR}^{\sf cov}(\gamma;\lambda,R_{\sf S}) \!\leq \!P_{\sf SIR}^{\sf cov, b}(\gamma; \lambda, R_{\sf S}, (m!)^{-\frac{1}{m}})
    \end{align}
    where
    \begin{align}
        \nonumber
        &P_{\sf SIR}^{\sf cov, b}(\gamma; \lambda, R_{\sf S}, \kappa)  
        \\ \nonumber
        &= 2\pi\lambda\frac{R_{\sf S}}{R_{\sf E}}\int_{R_{\rm min}}^{R_{\rm max}}  \Bigg( p_{\sf L}(r) \sum_{\ell=1}^{m} \binom{m}{\ell}(-1)^{\ell+1 }\mathcal{L}_{{  I}_{r}}\left(\frac{\ell m\kappa \gamma {r}^{\alpha_{\sf L}}}{G(r)} \right)  
        \\  \label{eq:Pcov_sandwich}
        &\quad + \big(1-p_{\sf L}(r)\big)\mathcal{L}_{{  I}_{r} }\left(\frac{\gamma r^{\alpha_{\sf N}}}{G(r)}\right)\Bigg)  r {\rm d} r,
    \end{align}
    \begin{proof}
       See Appendix~\ref{app:SandwichBounds}.
    \end{proof}
\end{theorem}

\begin{figure} 
    \centering 
    \includegraphics[width=0.95\columnwidth]{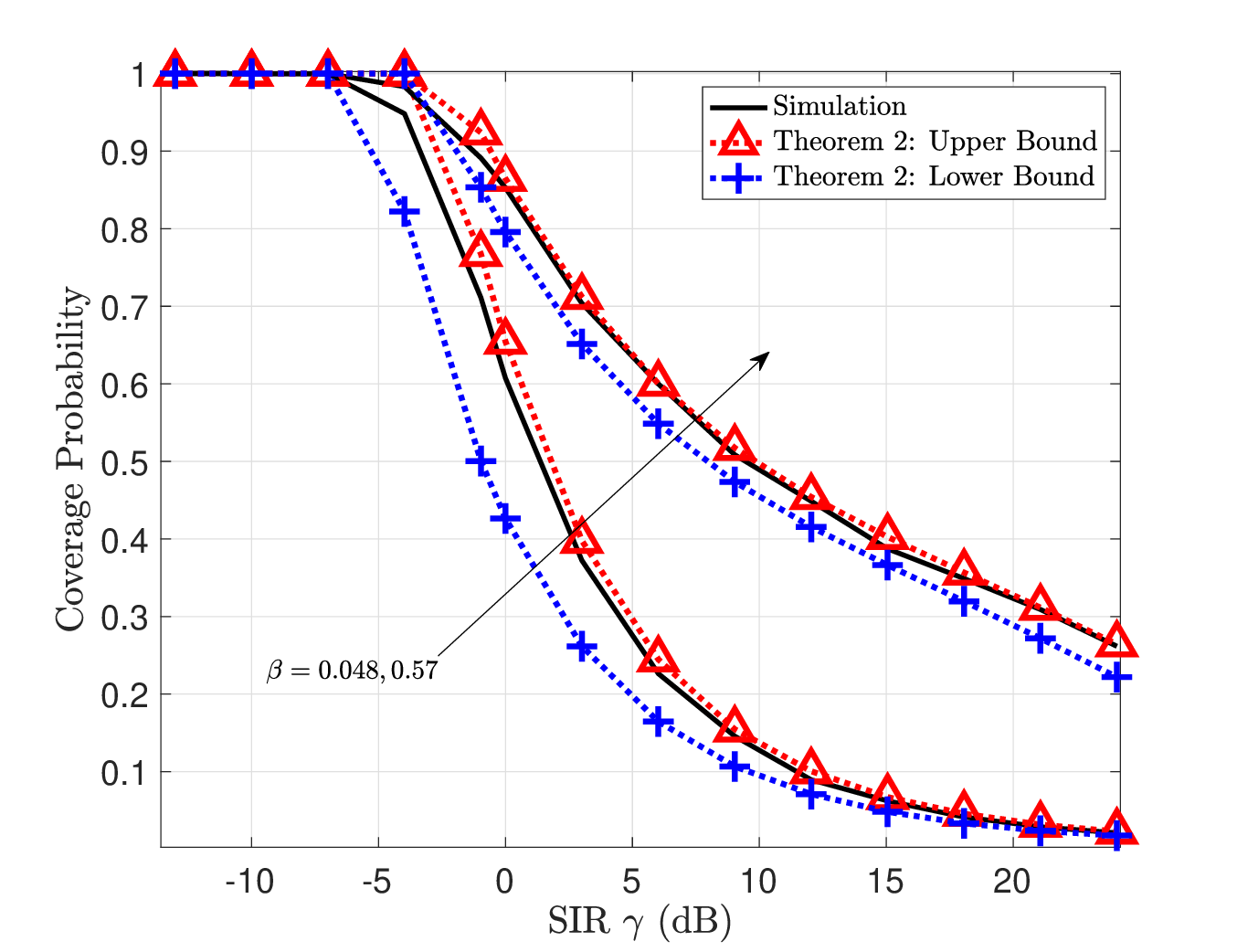}
    \caption{The numerical coverage probability and analytical bounds versus the target SIR $\gamma$ for $\alpha_{\sf L} = 2$, $\alpha_{\sf N} =3$, $h = 700$ km,  $m=3$,  $K=10$, and $\beta \in \{0.048, 0.57\}$.}
    \label{fig:SandwichBounds}
\end{figure}

We remark that the derived bounds in Theorem~\ref{thm:SandwichBounds} reduce to the exact coverage expression derived in Theorem~\ref{thm:Pcov_exact} when $m= 1$.
In addition, Fig.~\ref{fig:SandwichBounds} shows the validation of the derived bounds for the same simulation environment as that of Fig.~\ref{fig:Pcov_exact}.
As observed in Fig.~\ref{fig:SandwichBounds}, the upper and lower bounds are reasonably tight and capture the overall trend of the numerical coverage probability for the considered SIR. 
\begin{proposition}\label{prop:Pcov_approx}
    The coverage bound $P_{\sf SIR}^{\sf cov, b}(\gamma, \lambda, R_{\sf S}, \kappa)$ in Theorem~\ref{thm:SandwichBounds} is considered to be a good approximation for $1\leq \kappa \leq (m!)^{-\frac{1}{m}}$:
    \begin{align}
        P_{\sf SIR}^{\sf cov} \approx P_{\sf SIR}^{\sf cov, b}(\gamma, \lambda, R_{\sf S}, \kappa) \quad 1\leq \kappa \leq (m!)^{-\frac{1}{m}}
    \end{align}
    \begin{proof}
        The proof is straightforward from Theorem~\ref{thm:SandwichBounds}.
    \end{proof}
\end{proposition}

Accordingly, in the following subsection, we derive a closed-form expression for a special case by utilizing the bounds in Theorem~\ref{thm:SandwichBounds} to further draw  analytical insights.

\subsection{Simplified LOS Probability and Beam Gain}

In this subsection, we simplify the general LOS probability function $p_{\sf L}(r)$ as a step function. 
Such an approach has been often used for analyzing  wireless networks by explicitly incorporating the shadowing propagation \cite{bai2014coverageTWC, chae2023performance}.
The  LOS probability step function is defined as
\begin{align}
    \label{eq:P_LOS_step}
    p_{\sf L}(r) =  \mathbbm{1}{\{R_{\rm min} \leq  r < R_{\rm los}\}}.
\end{align}
This simplified LOS probability model \eqref{eq:P_LOS_step} indicates that only the satellites whose distance from the typical user is less than $R_{\rm los}$ experience LOS propagation.
Similarly, we also simplify the distance-based beam gain as 
\begin{align}
    \label{eq:G_step}
    G(r) = G_{\sf L}\mathbbm{1}{\{R_{\rm min} \!\leq\!  r \!< \!R_{\rm los}\}} \!+\!  G_{\sf N}\mathbbm{1}{\{R_{\rm los} \!\leq \!r\! \leq\! R_{\rm max} \}}.
\end{align}
Leveraging the simplified functions, we provide the  tractable closed-form coverage probability for the downlink LEO satellite network with shadowing channels. 

Based on Proposition~\ref{prop:Pcov_approx}, we derive the approximated coverage probability for the considered LOS probability and beam gain model in closed form in Proposition~\ref{prop:Pcov_approx_step} for a special case where $\alpha_{\sf N}/\alpha_{\sf L} = 2$.


\begin{proposition}\label{prop:Pcov_approx_step}
In the interference-limited regime, an approximation of the coverage probability with \eqref{eq:P_LOS_step}  and \eqref{eq:G_step}  for $\gamma > 0$ dB and  $\alpha_{\sf N}/\alpha_{\sf L} =2$  can be derived in closed form as in \eqref{eq:Pcov_approx_step}, which is at the top of the next page,
 \begin{figure*}
   \begin{align}
        \nonumber
        P^{\sf cov,  a}_{{\sf SIR}} (\gamma;\lambda,  R_{\sf S}, \kappa,\epsilon) 
        &= \frac{1}{4}\sqrt{\frac{c\pi}{\rho_{\sf L}(\frac{\gamma}{G_{\sf N}}, \alpha_{\sf N};\epsilon)}}\left(\frac{1}{3}e^{-c\Psi_{\sf 1}(R_{\rm los}^2,\frac{\gamma}{G_{\sf N}},\alpha_{\sf N})} -\frac{1}{3}e^{-c\Psi_{\sf 1}(R_{\rm max}^2,\frac{\gamma}{G_{\sf N}},\alpha_{\sf N})} +e^{-c\Psi_{\sf 2}(R_{\rm los}^2,\frac{\gamma}{G_{\sf N}},\alpha_{\sf N})} -e^{-c\Psi_{\sf 2}(R_{\rm max}^2,\frac{\gamma}{G_{\sf N}},\alpha_{\sf N})}\right) 
        \\\nonumber
        &\quad +\sum_{\ell=1}^{m}\! \binom{m}{\ell}\!{(-1)^{\ell+1 }}\!\Bigg[ \!\frac{\sqrt{c\pi} \rho_{\sf N}(z, \alpha_{\sf L};\epsilon)}{4\rho_{\sf L}(z, \alpha_{\sf L};\epsilon)^{\frac{3}{2}}}\!\left(\! \frac{1}{3}e^{-c\Psi_{\sf 1}\!(R_{\rm los},z,\alpha_{\sf L})}\! -\!\frac{1}{3}e^{-c\Psi_{\sf 1}\!(R_{\rm min},z,\alpha_{\sf L})} \!+ \!e^{-c\Psi_{\sf 2}\!(R_{\rm los},z,\alpha_{\sf L})} \!-\! e^{-c\Psi_{\sf 2}\!(R_{\rm min},z,\alpha_{\sf L})}\!\right) 
        \\ \label{eq:Pcov_approx_step}
        &\quad + \frac{e^{-c \Psi_{\sf 1}\!(R_{\rm min},z,\alpha_{\sf L})}\!-\!e^{-c\Psi_{\sf 1}\!(R_{\rm los},z,\alpha_{\sf L})}}{\rho_{\sf L}(z, \alpha_{\sf L};\epsilon)}\bigg|_{z=\frac{\ell m \gamma \kappa}{G_{\sf L}}}\!\Bigg].
   \end{align}
       \noindent\rule{\textwidth}{0.5pt}
\end{figure*}
    where $ c = \frac{\pi \lambda R_S}{R_E}$ and 
    \begin{align}
        \Psi_{\sf 1}(R,z,\alpha) &= R\big(\rho_{\sf N}(z,\alpha;\epsilon) + R \rho_{\sf L}(z,\alpha;\epsilon)\big),
        \\
          \Psi_{\sf 2}(R,z,\alpha) &=   \frac{4}{3}\Psi_{\sf 1}(R,z,\alpha) + \frac{1}{12}\frac{\rho_{\sf N}(z,\alpha;\epsilon)^2}{\rho_{\sf L}(z,\alpha;\epsilon)},
    \end{align}
    with 
    \begin{align}
        \nonumber
        \rho_{\sf L}( x,\alpha;\epsilon) &\!=\!  \left(\frac{xG_{\sf L}}{m}\right)^{\frac{2}{\alpha_{\sf L}}}\!\!\!\int_{\left(\frac{xG_{\sf L}}{m}\right)^{-\frac{2}{\alpha_{\sf L}}}\frac{R^2_{\rm min}}{(\epsilon R_{\rm max})^{\frac{2\alpha}{\alpha_{\sf L}}}}}^{\left(\frac{xG_{\sf L}}{m}\right)^{-\frac{2}{\alpha_{\sf L}}}\frac{R^2_{\sf L}}{(R_{\rm min}/\epsilon)^{\frac{2\alpha}{\alpha_{\sf L}}}} } \!\! 1\!-\!\frac{1}{ \left(1\!+\!u^{-\frac{\alpha_{\sf L}}{2}}\right)^{m}} {\rm d}u, 
        \\\nonumber
        \rho_{\sf N}(x,\alpha;\epsilon) &\!=\! (xG_{\sf N})^{\frac{2}{\alpha_{\sf N}}}\!\!\!\int_{(xG_{\sf N})^{-\frac{2}{\alpha_{\sf N}}}\frac{R^2_{\sf L}}{(\epsilon R_{\rm max})^{\frac{2\alpha}{\alpha_{\sf N}}}}}^{(xG_{\sf N})^{-\frac{2}{\alpha_{\sf N}}}\frac{R^2_{\rm max}}{(R_{\rm min}/\epsilon)^{\frac{2\alpha}{\alpha_{\sf N}}}} }  \!\!1\!-\! \frac{1}{1\!+\!u^{-\frac{\alpha_{\sf N}}{2}}}  {\rm d}u.
    \end{align}
    Here,  $1\leq \kappa \leq (m!)^{-\frac{1}{m}}$ and $ \max(\frac{R_{\rm min}^{\frac{\alpha + \alpha_{\sf L}}{2\alpha}}}{R_{\rm max}^{\frac{1}{2}}R_{\rm los}^{\frac{\alpha_{\sf L}}{2\alpha}}}, \frac{R_{\rm min}^{\frac{1}{2}}R_{\rm los}^{\frac{\alpha_{\sf N}}{2\alpha}}}{R_{\rm max}^{\frac{\alpha + \alpha_{\sf N}}{2\alpha}}}) \leq \epsilon \leq 1$ for $\alpha \in\{\alpha_{\sf L}, \alpha_{\sf N}\}$.
    \begin{proof}
     See Appendix~\ref{app:Pcov_approx_step}.
    \end{proof}
\end{proposition}
In addition, from the approximated coverage probability derived in Proposition~\ref{prop:Pcov_approx_step}, we have the following lower bound:
\begin{corollary}\label{cor:Pcov_lb_step}
    \normalfont The lower bound of the exact coverage probability for $\gamma > 0$ dB in the considered model is obtained as
    \begin{align}
        \label{eq:Pcov_lb_step}
         P^{\sf cov, lb}_{\sf SIR} (\gamma;\lambda, R_{\sf S}) = P^{\sf cov, \sf a}_{{\sf SIR}} (\gamma;\lambda, R_{\sf S},1,1).
    \end{align}
    \begin{proof}
        See Appendix~\ref{app:Pcov_lb_step}.
    \end{proof}
\end{corollary}

Fig.~\ref{fig:ApproxCoverage} shows the simulation coverage probability and analytical approximation in \eqref{eq:Pcov_approx_step}.
It is observed that the approximation provides reasonable accuracy with  similar trend as the simulation results.
In addition,  smaller $R_{\rm los}$ means denser urban with a less number of LOS satellites. 
Accordingly, smaller $R_{\rm los}$ results in lower interference as  interfering satellites are more likely to experience NLOS propagation.
The coverage probability in Fig.~\ref{fig:ApproxCoverage} corresponds to such network intuition.
In this regard, using the derived approximation for network analysis is considered to be valid.

\begin{figure} 
    \centering 
    \includegraphics[width=0.95\columnwidth]{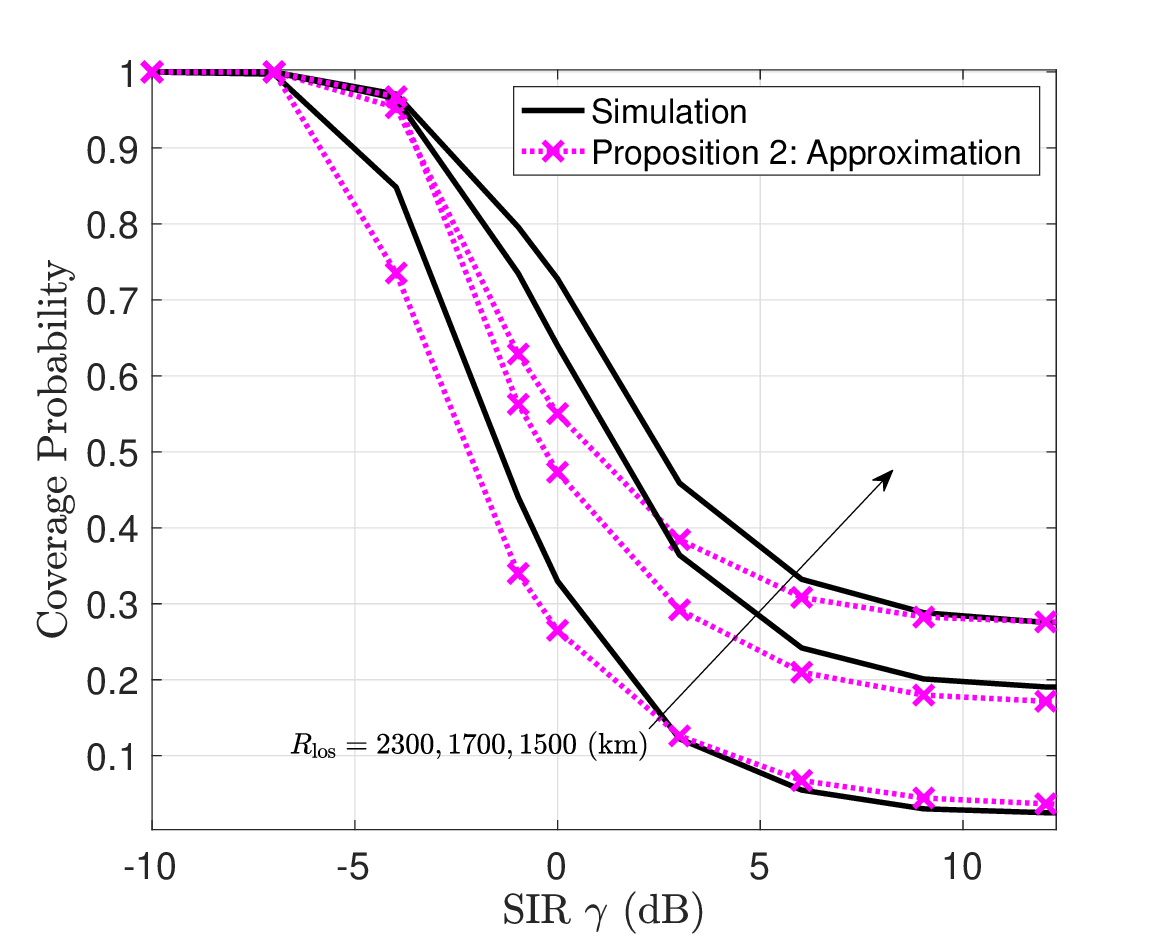}
    \caption{The simulation coverage probability and analytical approximation versus the target SIR $\gamma$ for the simplified LOS probability and beam gain functions with $\alpha_{\sf L} = 2$, $\alpha_{\sf N} =4$, $h = 700$ km,  $m=3$,  $K=10$, $R_{\sf los} \in \{1500, 1700, 2300\}$ km, $\kappa =(m!)^{-\frac{1}{m}}$, and $\epsilon = 0.6$}
    \label{fig:ApproxCoverage}
\end{figure}

For instance, in Proposition~\ref{prop:Pcov_approx_step}, consider $c\to \infty$, i.e., extremely-dense satellite networks. 
Then \eqref{eq:Pcov_approx_step} becomes $P_{\sf SIR}^{\sf cov, A} \to 0$ since $\lim_{x\to \infty}\sqrt{ax}\cdot e^{-bx}\to 0$ for $a,b>0$.
This analysis aligns with the intuition that with sufficiently many satellites, deploying more satellites will only deteriorate the coverage probability by causing additional interference with marginal improvement in the signal power of the associated satellite.
This is different from the analytical results in the terrestrial network in which the coverage probability is independent to the BS density in the interference-limited regime \cite{andrews:tcom:11}.
The key reason for such difference comes from the minimum distance between the transmitter and receiver which is zero for the terrestrial network and $R_{\rm min}$ for the considered satellite network.
This observation implies that there should be an optimal density.

To verify this insight  in a special case,  we derive the lower bound in a simpler form for a case where channels are homogeneous and Rayleigh fading, and subsequently identify the optimal density.

\begin{theorem}\label{thm:lowerbound_homo}
    When all channels are homogeneous and follows Rayleigh fading, the coverage probability for $\gamma >0$ dB in the interference-limited regime is lower bounded by
        \begin{align}
            \label{eq:lowerbound_homo}
            &P^{\sf cov, hm, lb}_{{\sf SIR}} (\gamma;\lambda, R_{\sf S})
            \\\nonumber
            &=\frac{1}{\rho^{\sf hm}({\gamma};\alpha)}\left(e^{-\frac{\pi\lambda  R_{\sf S}}{R_{\sf E}}\rho^{\sf hm}({\gamma};\alpha)R_{\rm min}^2} - e^{-\frac{\pi\lambda R_{\sf S}}{R_{\sf E}}\rho^{\sf hm}({\gamma};\alpha)R_{\rm max}^2}\right),
        \end{align}
        where
        \begin{align}
            \rho^{\sf hm}(\gamma;\alpha) = \gamma^{\frac{2}{\alpha}}\int_{\gamma^{-\frac{2}{\alpha}}\frac{R^2_{\rm min}}{R_{\rm max}^2}}^{\gamma^{-\frac{2}{\alpha}}\frac{R^2_{\rm max}}{R_{\rm min}^2} }  1-\frac{1}{1+u^{-\frac{\alpha}{2}}} {\rm d}u. 
        \end{align}
    \begin{proof}
        See Appendix~\ref{app:lowerbound_homo}.
    \end{proof}
\end{theorem}
We remark that this probability is different from the one in \cite{park2022tractable} since  the nearest satellite association rule is considered in \cite{park2022tractable}.
In the following theorem, we derive the optimal density that maximizes the lower bound in \eqref{eq:lowerbound_homo}.
\begin{theorem}
    \label{thm:OptimalDensity}
    The optimal density that maximizes the coverage lower bound in \eqref{eq:lowerbound_homo} is derived as
    \begin{align}
        \label{eq:OptimalDensity}
        \lambda^\star_{\sf lb} = \frac{R_{\sf E}\log \left(1+\frac{2R_{\sf E}}{h}\right)}{2\pi (R_{\sf E}+h)hR_{\sf E} \rho^{\sf hm}({\gamma};\alpha)}.
    \end{align}
    \begin{proof}
        See Appendix~\ref{app:OptimalDensity}.
    \end{proof}
\end{theorem}
The derived optimal density in Theorem~\ref{thm:OptimalDensity} allows us to approximate the desirable satellite density for any given satellite altitude  $h$ and target SIR $\gamma$.
Fig.~\ref{fig:OptimalDensity} shows the average number of satellites in the spherical cap $\mathcal{A}$ with the derived optimal density \eqref{eq:OptimalDensity} and the one in \cite{park2022tractable}  versus satellite altitude $h$ for $\gamma \in \{0, 3, 5\}$ dB and $\alpha = 3$. 
The optimal density derived under the strongest association rule is smaller than the one derived under the nearest association rule. 
This is because the network with the nearest association  policy is sub-optimal so that it requires more satellites to be deployed to improve the desired signal.
This suggests that depending on the network's capability in the satellite association accuracy as well as the operating altitude and target SIR, deploying policy should change.
The derived densities in Theorem~\ref{thm:OptimalDensity} and in \cite{park2016optimal} can provide guidelines on the optimal satellite density.

\begin{figure} 
    \centering 
    \includegraphics[width=0.98\columnwidth]{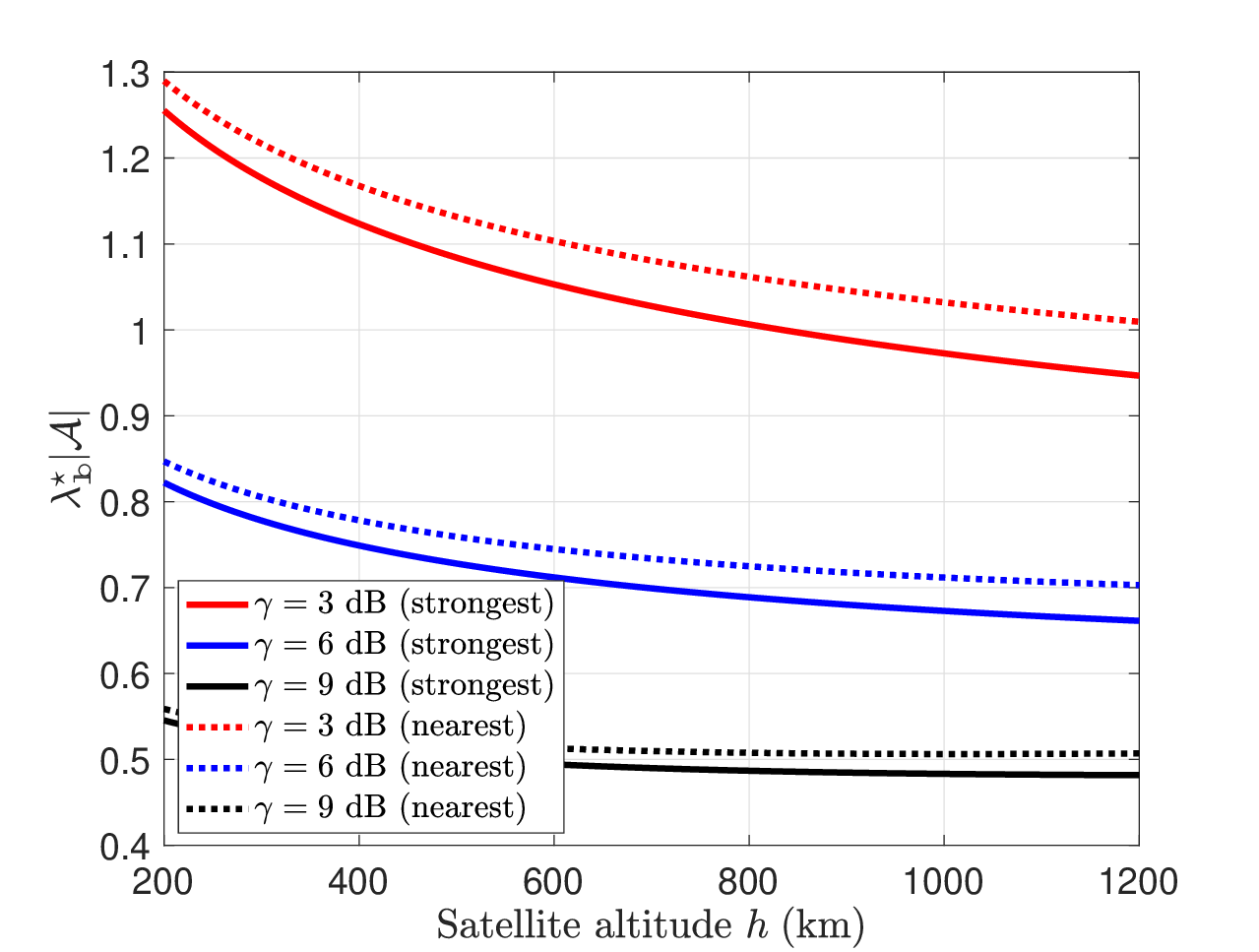}
    \caption{The average number of satellites in the spherical cap with the derived optimal density for strongest and nearest association rules versus satellite altitude $h$ for for the simplified LOS probability and beam gain functions with  target SIR $\gamma \in \{0, 3, 5\}$ dB and pathloss exponent $\alpha = 3$.}
    \label{fig:OptimalDensity}
\end{figure}


\section{Numerical Study}
\label{sec:numerical}

In this section, we further study the LEO satellite networks based on the derived analytical results via numerical evaluation.
In the evaluation, we set the satellite altitude as $h=700$ km and the Earth radius as $R_{\sf E} = 6371$ km.
We use the exponential LOS probability in \eqref{eq:Plos_exp} and Bessel beam gain in \eqref{eq:bessel} with $G_{\rm max} = 20$ dB.
The considered setting holds throughout this section  unless mentioned otherwise.

We first show the coverage probabilities for the strongest association and nearest association rules. 
In Fig.~\ref{fig:CovSIR}, we consider $h=700$ km,  $m=3$, $\alpha_{\sf L} = 3$, $\alpha_{\sf N} = 4$,  $\beta \in \{0.048, 0.2, 0.57\}$, $\theta_{\rm 3dB} = 10^\circ$, and $K=10$.
Fig.~\ref{fig:CovSIR} shows that there is a significant gap between the strongest association and nearest association cases, which underscores the importance of considering the random shadowing effect in the satellite network modeling.
Regarding the coverage performance, the strongest satellite association provides a higher coverage probability than the nearest association, which corresponds to a general intuition. 
As $\beta$ increases, i.e., denser urban, the coverage probability also increases for the strongest association case, which is also verified in Fig.~\ref{fig:Pcov_exact} for $\alpha_{\sf L} = 2$ and $\alpha_{\sf N} = 3$.

\begin{figure} 
    \centering 
    \includegraphics[width=1\columnwidth]{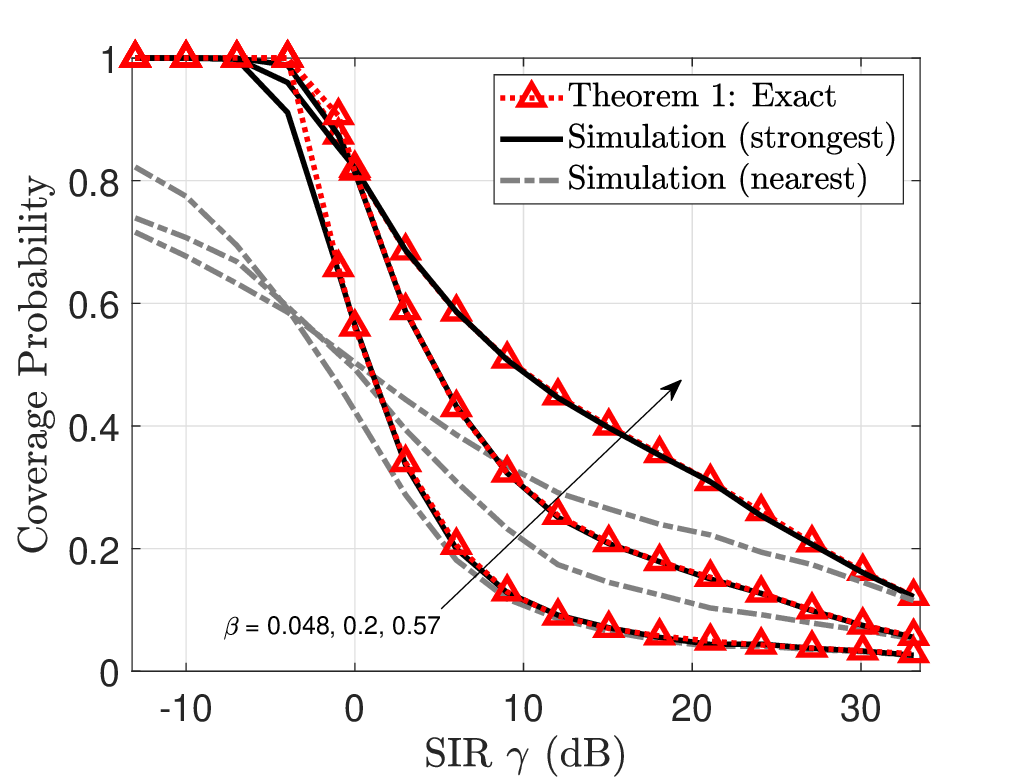}
    \caption{The coverage probability according to the target SIR $\gamma$  for $h=700$ km,  $m=3$, $\alpha_{\sf L} = 3$, $\alpha_{\sf N} = 4$,  $\beta \in \{0.048, 0.2, 0.57\}$,  $\theta_{\rm 3dB} = 5^\circ$, and $K=10$.}
    \label{fig:CovSIR}
\end{figure}

In Fig.~\ref{fig:PcovVsDensity}, we present the coverage probability in Theorem~\ref{thm:Pcov_exact} with respect to the satellite density $\lambda$.
As discussed in the observation from Proposition~\ref{prop:Pcov_approx_step}, there exists an optimal density for each target SIR.
A noticeable observation is that the optimal density decreases as the target SIR increases while achieving a lower coverage probability.
This phenomenon occurs because the interference from non-associated satellites significantly deteriorates the coverage performance, and thus, the higher target SIR requires a smaller number of visible satellites.
Consequently, we need to carefully choose the density when deploying a satellite network to maximize the coverage performance depending on the target operating SIR. 
In addition, it is shown that as the environment becomes more urban (higher $\beta$), the optimal density increases, and the coverage probability becomes less sensitive to the density.
This is because as $\beta$ increases, the LOS probability decreases, thereby reducing the interference and allowing more satellites to be deploy to increase the desired signal power.
In this regard, we need to be careful in deploying  satellites and planning their orbits so that the different average number of satellites that are close to optimal serve different regions.

\begin{figure} 
    \centering 
    \includegraphics[width=0.95\columnwidth]{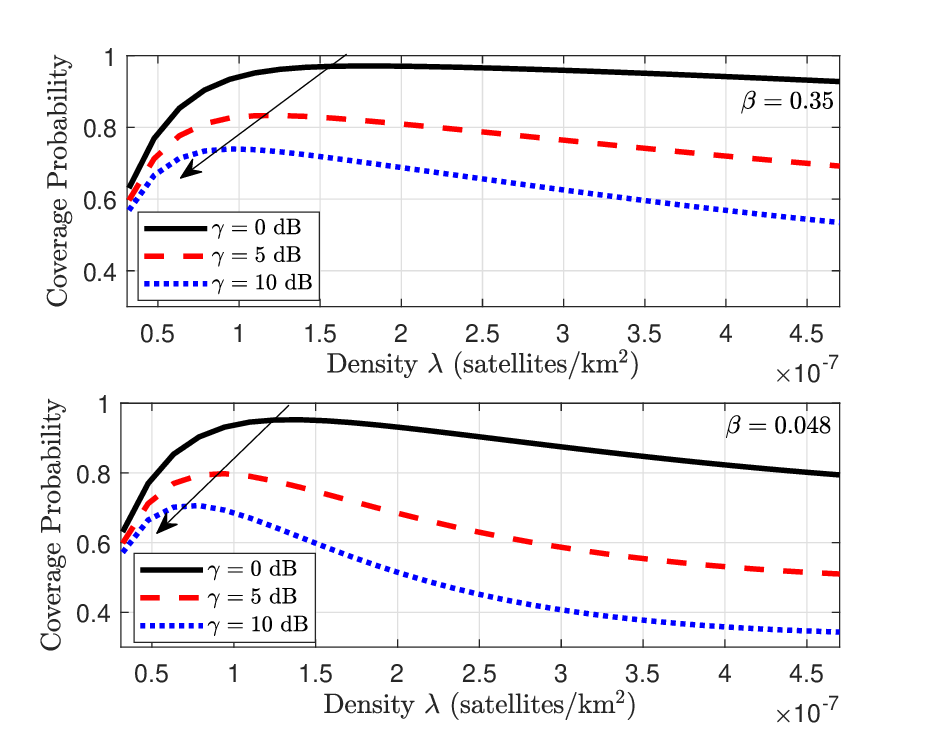}
   \caption{The coverage probability versus the satellite density with  $h=700$ km and the satellite altitude with $K = 10$ for $\alpha_{\sf L}=2.5$, $\alpha_{\sf N}=4$, $m=3$,  $\beta \in \{0.048, 0.35\}$ dB, $\theta_{\rm 3dB} = 10^\circ$, and $\gamma \in \{0,5,10\}$ dB.}
   \label{fig:PcovVsDensity}
\end{figure}

Fig.~\ref{fig:PcovVsAltitude}(a) shows the coverage probability in Theorem~\ref{thm:Pcov_exact} with respect to the satellite altitude $h$ for $K = 10$, $\alpha_{\sf L}=2.5$, $\alpha_{\sf N}=4$, $m=3$,  $\beta \in \{0.048, 0.35\}$ dB, $\theta_{\rm 3dB} = 10^\circ$, and $\gamma \in \{0,5,10\}$ dB.
It is interesting to note that the altitude can be divided into three different regimes: beam gain, interference, and pathloss-dominant regimes.
When the altitude is low, increasing the altitude improves the coverage probability by offering more opportunity to the receiver to experience higher beam gain as shown in Fig.~\ref{fig:PcovVsAltitude}(b).
This means that as the altitude increases, the beam coverage also increases and the strongest satellite is more likely to provide higher beam gain than the interfering satellites.
When the altitude is medium, however, the interfering satellites also begin to interfere with the higher beam gain, which decreases the coverage probability.
Finally, when the altitude is high, although the beam gain effect is balanced between the associated satellite and interfering satellites, the pathloss of the interfering satellites becomes much severe as they are more likely to experience NLOS propagation than the associated satellite.
Consequently, the coverage probability marginally increases with the altitude in the high altitude regime.
Based on the observation, lowering the satellite altitude close to the optimal point is desirable, which aligns with the motivation of LEO satellite systems.

\begin{figure}[!t]\centering
    \begin{subfigure}[Coverage probability]{\resizebox{0.95\columnwidth}{!}{\includegraphics{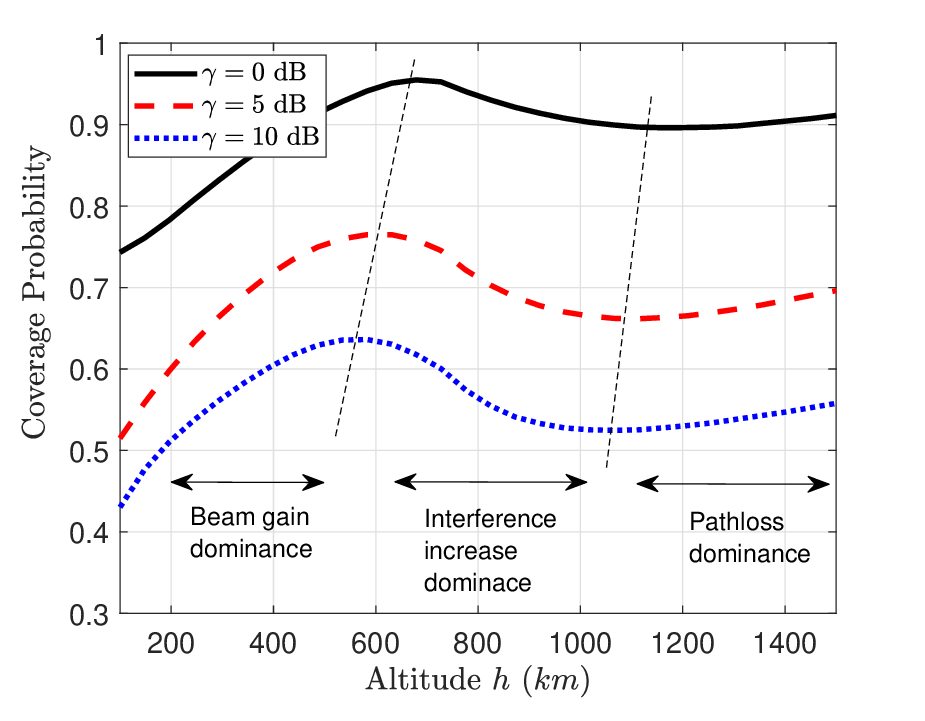}}}
    \end{subfigure}
    \begin{subfigure}[Normalized beam gain contour at different altitudes $h$]{\resizebox{0.95\columnwidth}{!}{\includegraphics{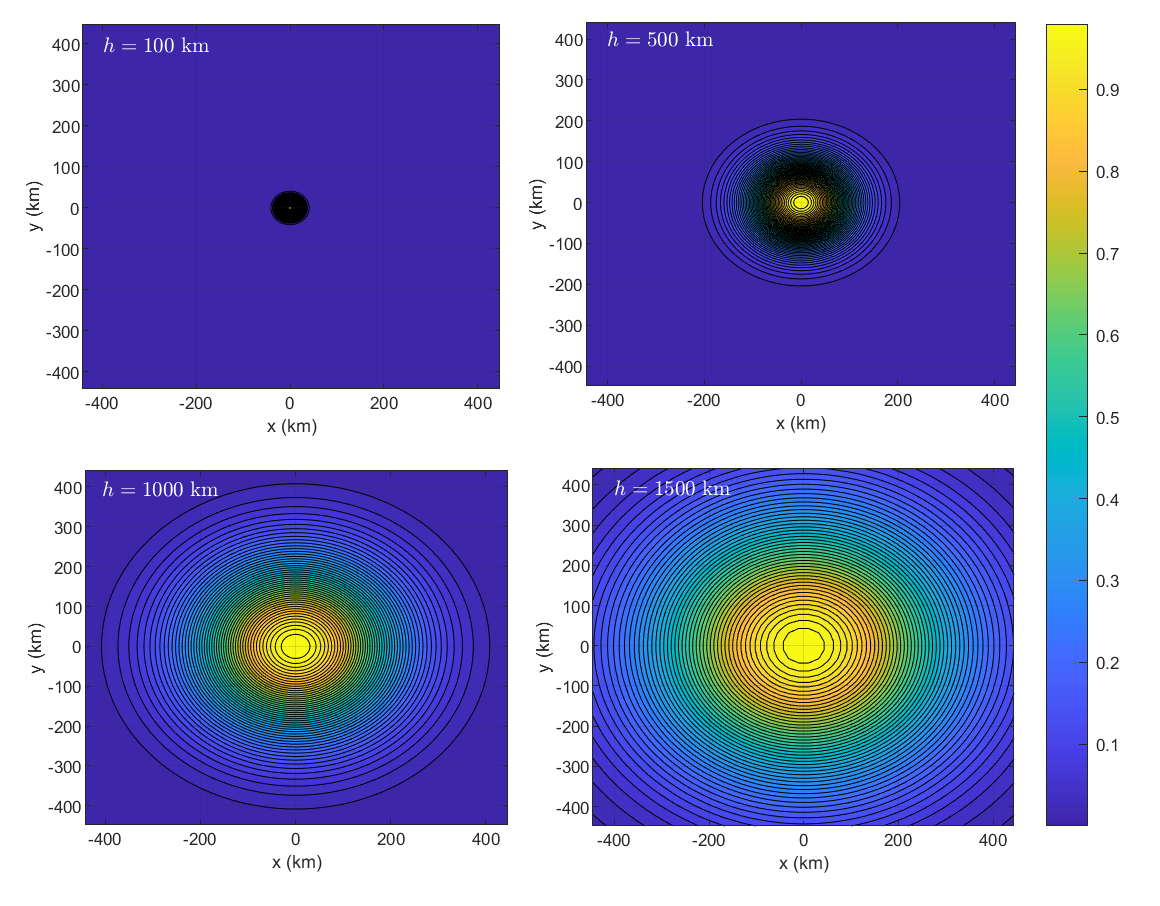}}}
    \end{subfigure}
    \caption{The coverage probability and beam gain contour on the Earth surface according to the satellite altitude for $\alpha_{\sf L}=2.5$, $\alpha_{\sf N}=4$,  $m=3$,  $\theta_{\rm 3dB} = 10^\circ$, and $\gamma\in \{0,5,10\}$ dB.}
    \label{fig:PcovVsAltitude}
\end{figure}

Finally, we  compare the LEO satellite modeling based on the PPP, BPP~\cite{okati2020downlink}, Walker star constellation, and collected Starlink satellite distribution data~\cite{park2022tractable} for the considered shadowing system.
For comparison, we first choose the Starlink satellites whose altitude is in $[400, 450]$ km from the collected data and randomly assign fractional frequency reuse with $20$ different frequencies. 
Then we compute the average number of visible Starlink satellites on the frequency of interest which turns out to be about $4.3$  satellites.
We use it to determine the satellite density for the PPP case, the number of visible satellites for the BPP case, and the average number of visible satellites for the Walker star constellation case (for the Walker constellation, it corresponds to 60 orbital planes with 25 satellites per orbital plane).
In addition, we limit the elevation angle of the visible satellites from the receiver tangential plan by extracting the minimum elevation angle of the Starlink satellites which is  $\sim 25^\circ$, not $0^\circ$, 
Fig.~\ref{fig:comparison} shows the coverage probabilities from the PPP, BPP, Walker, and Starlink cases for $\alpha_{\sf L}=3$, $\alpha_{\sf N}=4$, $h=425$ km, $m=3$, $\beta = 0.35$, and $\theta_{\rm 3dB} = 5^\circ$ dB.
The coverage probabilities reasonably match with each other, and this demonstrates the validity of the PPP-based satellite distribution modeling used in this work.
In this regard, we can conclude that the coverage analyses provided in this paper can properly guide the LEO satellite network design.

\begin{figure}[!t]\centering
    \begin{subfigure}[Coverage probability]{\resizebox{0.95\columnwidth}{!}{\includegraphics{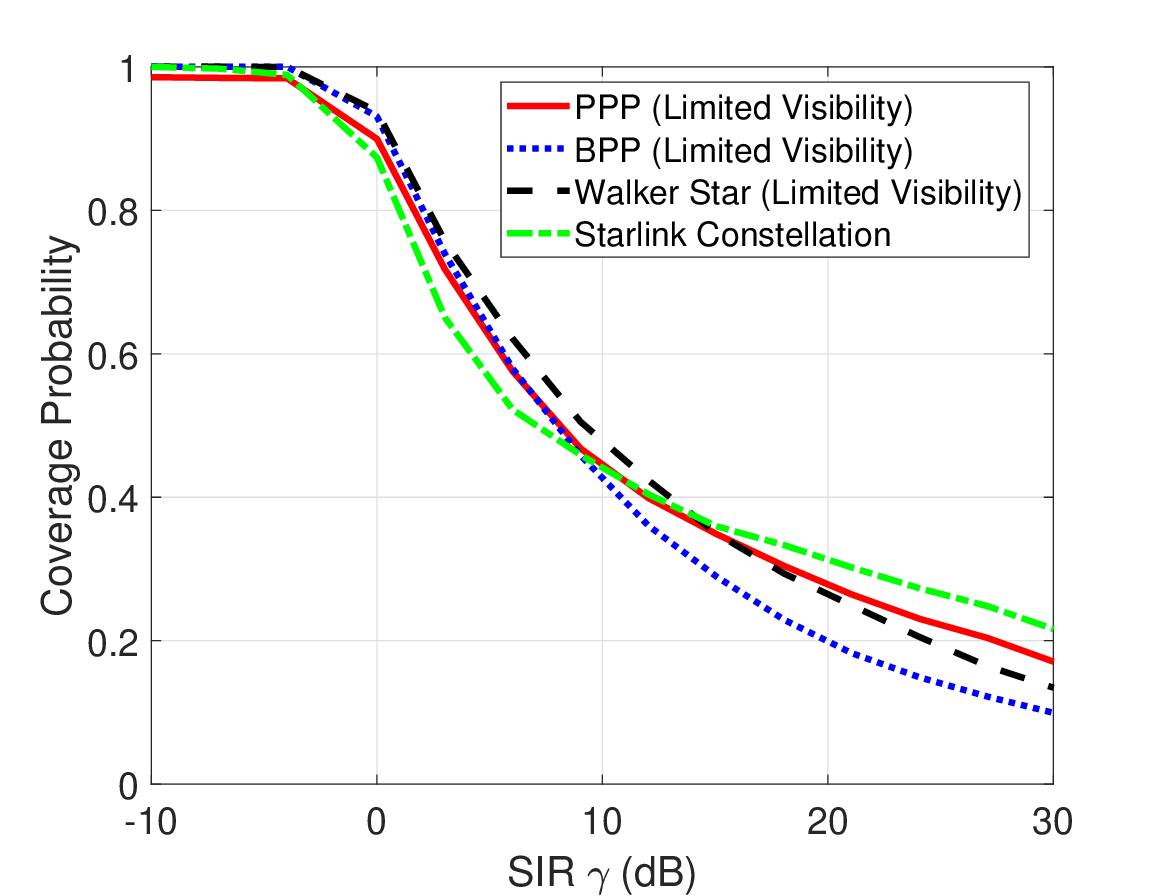}}}
    \end{subfigure}
    \begin{subfigure}[Walker star constellation snapshot]{\resizebox{0.95\columnwidth}{!}{\includegraphics{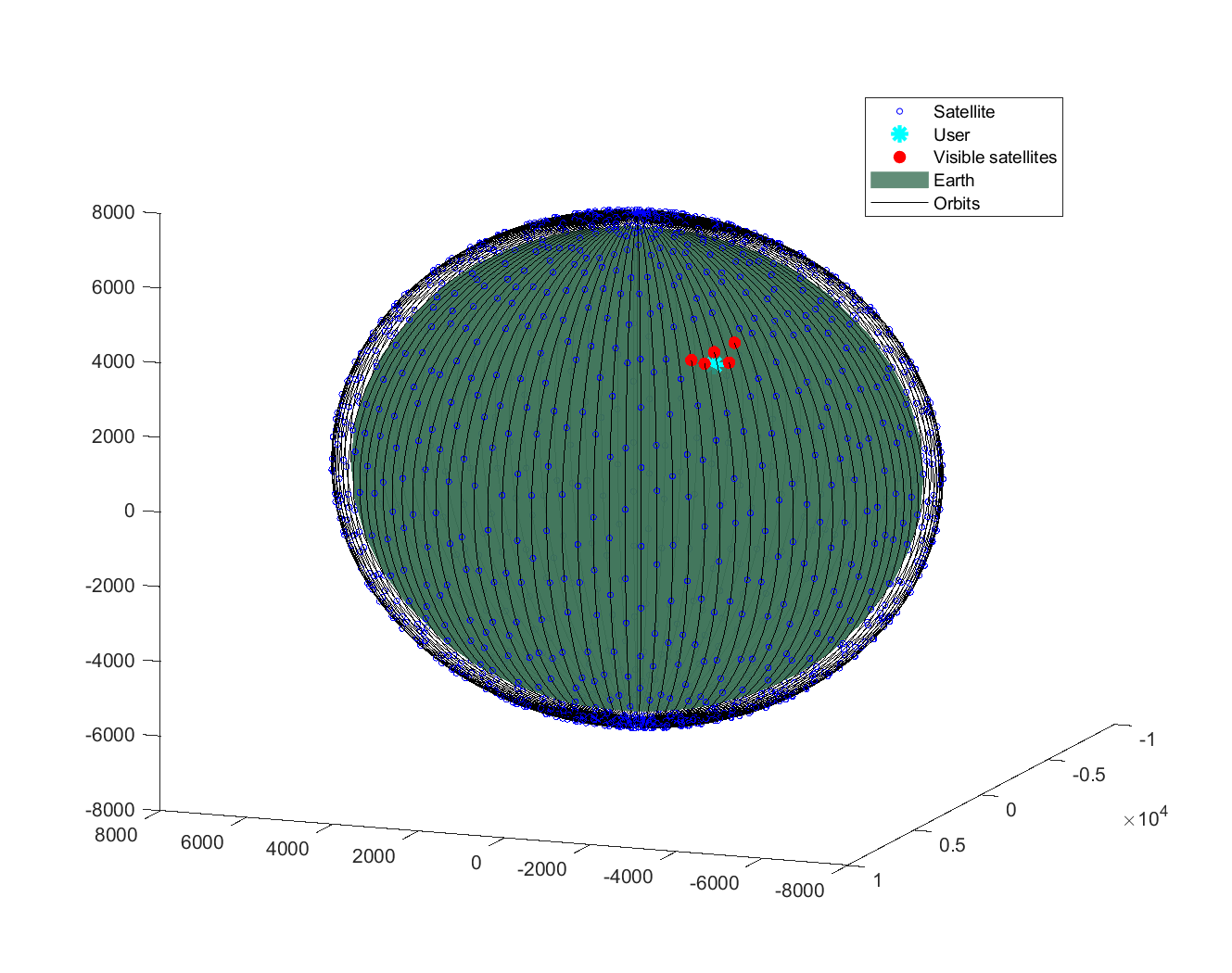}}}
    \end{subfigure}
    \caption{The coverage probabilities from PPP, BPP,  collected Starlink constellation, and Walker Star constellation for $\alpha_{\sf L}=3$, $\alpha_{\sf N}=4$, $h=425$ km, $m=3$, $\beta = 0.35$, and $\theta_{\rm 3dB} = 5^\circ$ dB.}
    \label{fig:comparison}
\end{figure}



\section{Conclusion}
\label{sec:conclusion}
This paper analyzed satellite network coverage with shadowing effects using a PPP-based model. We derived an analytical expression for coverage probability under a strongest satellite association scenario, incorporating distance-dependent shadowing. By further deriving sandwich bounds, we obtained a closed-form coverage probability for a simplified shadowing model. Our analysis also identified the optimal satellite density as a function of network parameters in non-shadowing scenarios. The results show that shadowing improves coverage probability by reducing interference. We observed that denser urban environments are less sensitive to optimal satellite density, and an optimal altitude exists where beam gain and interference from wider beams are balanced. Our network model demonstrated similar coverage probability to BPP, Walker star constellation, and Starlink models. These findings offer a valuable framework for optimizing LEO satellite deployment strategies, accounting for shadowing effects, and provide guidance for enhancing network performance.

Future research could focus on extending our distance-dependent blockage model analysis framework to address spectrum sharing scenarios involving terrestrial networks \cite{park:arxiv:23, kim:arxiv:24,Yim:arxiv:24}.


\appendices

\section{Proof of Lemma \ref{lem:laplace}} \label{app:laplace}
Without loss of generality, we assume ${\bf x}_i = \bx_1$ and $r_i = r$ in this proof. 
Since $H_j$ is the Nakagami-$m$ fading for the LOS channel, the complementary cumulative distribution function (CCDF) of $H_j$ for LOS channels is given by
\begin{align}
    \label{eq:pdf_los}
	 \mathbb{P}[H_j^{\sf L}\geq x] 	&=  e^{-m x}\sum_{k=0}^{m-1}\frac{(m x)^k}{k!}.
\end{align}
The CCDF of $H_j$ for NLOS channels which follow Rayleigh fading distribution is given as
\begin{align}
     \label{eq:pdf_nlos}
	 \mathbb{P}[H_j^{\sf N}\geq x] =  e^{-x}.
\end{align}
We compute the Laplace transform of the aggregated interference power as
\begin{align}
    \label{eq:condLapla} 
    &\mathcal{L}_{I_{r}}(s) 
    =\mathbb{E}\left[ e^{-s I_r} ~ \middle|~ \|{\bf x}_{1}-{\bf u}_1\|=r\right] 
    \\\nonumber
    &=\bbE \left[\prod_{{\bf x}_j\in \Phi_{{\rm I}(\bx_1)}} e^{-sG(r_j)H_j r_j^{-\alpha_j}} \middle|~ \|{\bf x}_{1}-{\bf u}_1\|=r\right]
    \\\label{eq:condLapla_sub1} 
    &\!\stackrel{(a)}{=} \bbE_{\Phi} \Bigg[\prod_{{\bf x}_j\in \Phi_{{\rm I}(\bx_1)}} \Bigg(p_{\sf L}(r_j)\bbE_{H^{\sf L}_j}\left[e^{-sG(r_j)H^{\sf L}_j r_j^{-\alpha_{\sf L}}}\right] 
    \\
    &\quad+\big(1-p_{\sf L}(r_j)\big)\bbE_{H^{\sf N}_j}\left[e^{-sG(r_j)H^{\sf N}_j r_j^{-\alpha_{\sf N}}}\right]\Bigg)
     \Bigg|~\|{\bf x}_{1}-{\bf u}_1\|=r \Bigg],
\end{align}
where $(a)$ comes from the PDFs of LOS annd NLOS probability distributions in \eqref{eq:LosNlos}.
Regarding the strongest satellite association policy, the interfering satellites can reside anywhere on the typical spherical cap $\cA$. 
Then from the probability generating functional (PGFL) of the PPP \cite{baccelli2006aloha,haenggi2009stochastic}, \eqref{eq:condLapla_sub1} further becomes
\begin{align}
    \nonumber
    &\exp\Bigg( -\lambda \int_{\bx_j\in \mathcal{A}}  \left(1-p_{\sf L}(r_j)\bbE_{H^{\sf L}_j}\left[e^{-sG(r_j)H^{\sf L}_j r_j^{-\alpha_{\sf L}}}\right] \right.
    \\\nonumber
    &\qquad \left.-\big(1-p_{\sf L}(r_j)\big)\bbE_{H^{\sf N}_j}\left[e^{-sG(r_j)H^{\sf N}_j r_j^{- \alpha_{\sf N}}}\right] \right)~  {\rm d}\bx_j \Bigg)
    \\\nonumber
    &= \exp\Bigg( -\lambda \int_{\bx_j\in \mathcal{A}}  \Bigg(1-p_{\sf L}(r_j)\frac{1}{\left(1+\frac{sr_j^{-\alpha_{\sf L}}G(r_j)}{m}\right)^{m}}
    \\\nonumber
    &\qquad -\big(1-p_{\sf L}(r_j)\big)\frac{1}{1+sr_j^{-\alpha_{\sf N}}G(r_j)} \Bigg)~  {\rm d}\bx_j \Bigg)
    \\\nonumber
    &\stackrel{(b)}{=} \exp\Bigg(-2\pi \lambda \frac{R_{\sf S}}{R_{\sf E}}\int_{R_{\rm min}}^{R_{\rm max}}  \Bigg(1-p_{\sf L}(v)\frac{1}{\left(1+\frac{sv^{-\alpha_{\sf L}}G(v)}{m}\right)^{m}} 
    \\\label{eq:LaplaRay}
    &\qquad -\big(1-p_{\sf L}(v)\big)\frac{1}{1+sv^{-\alpha_{\sf N}}G(v)} \Bigg)~  v{\rm d}v \Bigg),
\end{align}
where $(b)$ comes from $	\frac{\partial |\mathcal{A}_v|}{\partial v}=2\pi \frac{R_{\sf S}}{R_{\sf E}} v$.
\qed

\section{Proof of Theorem \ref{thm:Pcov_exact}} \label{app:Pcov_exact}
For the interference-limited regime where $I_{r_i} \gg \bar \sigma^2$, we can ignore the noise power and derive the coverage probability by using the SIR. 
We  compute the coverage probability for $\gamma > 0$ dB as 
\begin{align}
    \nonumber
    P^{\sf cov}_{{\sf SIR}} (\gamma;\lambda, R_{\sf S}) &=\bbP\left[ \bigcup_{\bx_i \in \Phi \cap \mathcal{A}} \bigg\{H_i \geq  \frac{r_i^{\alpha_i}\gamma I_{r_i}}{G(r_i)}\bigg\}\right]
    \\\nonumber
    & = \bbE\left[\mathbbm{1}\left\{\bigcup_{\bx_i \in \Phi \cap \mathcal{A}} \bigg\{H_i \geq  \frac{r_i^{\alpha_i}\gamma I_{r_i}}{G(r_i)}\bigg\}\right\}  \right]
    \\\nonumber
    & \!\stackrel{(a)}{=}\bbE\left[\sum_{\bx_i \in \Phi \cap \mathcal{A}}  \mathbbm{1}\bigg\{H_i \geq  \frac{r_i^{\alpha_i}\gamma I_{r_i}}{G(r_i)}\bigg\} \right]
    \\
    \label{eq:thm_proof_mid}
    & =\bbE_{\Phi}\left[\sum_{\bx_i \in \Phi \cap \mathcal{A}}  \bbP\left[H_i \geq  \frac{r_i^{\alpha_i}\gamma I_{r_i}}{G(r_i)} \bigg|r_i\right]  \right],
\end{align}
where $(a)$ follows from Lemma 1 in \cite{dhillon2012modeling} under the assumption of $\gamma > 0$ dB.
We note that when $\gamma \leq 0$ dB, it is effectively an upper bound.
Then, from Campbell-Mecke Theorem \cite{stoyan2013stochastic}, \eqref{eq:thm_proof_mid} further becomes
\begin{align}
    \nonumber
    &\int_{\bx_i\in\mathcal{A}}\bbP\left[H_i \geq \frac{ r_i^{\alpha_i}\gamma I_{r_i}}{G(r_i)} \bigg|r_i \right]\lambda  {\rm d}\bx_i 
    \\\nonumber
    &= \int_{\bx_i\in\mathcal{A}}\mathbb{E}_{ I_{r_i} } \Bigg[  p_{\sf L}(r_i)\bbP\left[H_i^{\sf L} \geq \frac{r_i^{\alpha_{\sf L}}\gamma  I_{r_i}}{G(r_i)}   \bigg|  I_{r_i}\right]
    \\\nonumber
    &\qquad+\big(1-p_{\sf L}(r_i)\big)\bbP\left[H_i^{\sf N} \geq \frac{r_i^{\alpha_{\sf N}}\gamma I_{r_i}}{G(r_i)}  \bigg|  I_{r_i}  \right] \bigg|  r_i\Bigg] \lambda {\rm d}{\bx_i}
    \\\nonumber
    &= \int_{\bx_i\in\mathcal{A}}\Bigg( p_{\sf L}(r_i) \mathbb{E}_{I_{r_i}} \left[ \sum_{k=0}^{ m  -1} \frac{ m^k\gamma^k{r_i}^{k\alpha_{\sf L} }}{k!{G^k(r_i)}}{ I}_{r_i}^k  e^{-\frac{m  \gamma{r_i}^{\alpha_{\sf L}} {I}_{r_i}}{G(r_i)}  }\middle |  r_i\right] 
    \\\nonumber
    &\qquad+  \big(1-p_{\sf L}(r_i)\big)\mathbb{E}_{I_{r_i}} \left[e^{-\frac{\gamma r^{\alpha_{\sf N}} I_{r_i}}{G(r_i)}}  \bigg|  r_i \right]  \Bigg) \lambda{\rm d}{\bx_i}
    \\\nonumber
    &\stackrel{(c)}{=}\int_{R_{\rm min}}^{R_{\rm max}}  \Bigg( p_{\sf L}(r)\sum_{k=0}^{ m -1}\frac{ (-m)^k\gamma^k{r}^{k\alpha_{\sf L} }}{k!{G^k(r)}}\left.{\frac{\d^k\mathcal{L}_{{  I}_{r}}\!(s)}{\d s^k}} 
    \right|_{\begin{subarray}{l}s= \frac{m\gamma  {r}^{\alpha_{\sf L}}}{G(r)}\end{subarray}}  
    \\\label{eq:ProofCovStrong}
    &\qquad+ \big(1-p_{\sf L}(r)\big)\mathcal{L}_{{  I}_{r} }\left(\frac{\gamma r^{\alpha_{\sf N}}}{G(r)}\right)\Bigg) 2\pi\lambda\frac{R_{\sf S}}{R_{\sf E}} r {\rm d} r, 
\end{align}
where $(c)$ is from  Lemma~\ref{lem:laplace}, applying the derivative property of the Laplace transform, i.e., $\mathbb{E}\left[X^{k} e^{-sX}\right]=(-1)^{k}\frac{\d^k\mathcal{L}_X(s)}{\d s^k}$, and  $\frac{\partial |\mathcal{A}_r|}{\partial r}=2\pi \frac{R_{\sf S}}{R_{\sf E}} r$.
This completes the proof.
\qed

\section{Proof of Theorem \ref{thm:SandwichBounds}} \label{app:SandwichBounds}
 We devote to proving the upper bound, since the lower bound is readily obtained from the former by choosing $\kappa=1$. 
        According to Appendix~\ref{app:Pcov_exact}, the coverage probability for $\gamma > 0$ dB in the interference-limited regime is given as
        \begin{align}
            \label{eq:Pcov_homo_recall}
            P^{\sf cov}_{{\sf SIR}} (\gamma;\lambda, R_{\sf S})=  \int_{\bx_i\in\mathcal{A}}\bbP\left[H_i \geq  \frac{r_i^{\alpha_i}\gamma I_{r_i}}{G(r_i)} \bigg|r_i \right]\lambda  {\rm d}\bx_i 
        \end{align}
        The CCDF for $H^{\sf L}$ can be represented in terms of  the lower incomplete gamma function as
        \begin{align}
            \mathbb{P}[H^{\sf L} > x] 
            &= 1- \frac{1}{\Gamma(m)}\int_{0}^{mx} t^{m-1}e^{-t}{\rm d}t.
        \end{align}
        From the Alzer's inequality \cite{alzer:97,lee2014spectral},  the incomplete Gamma function has an expression in the middle sandwiched between two inequalities: 
        \begin{align}
            \nonumber
            \left(1-e^{-m \kappa x}\right)^{m} \leq  \frac{1}{\Gamma(m)}\int_{0}^{mx} t^{m-1}e^{-t}{\rm d}t  \leq     \left(1-e^{-m  x}\right)^{m}.
        \end{align}
        Using this sandwich inequality, the CCDF for $H_1$ is upper and lower bounded by
        \begin{align}
            \label{eq:ccdf_H0_bound}
           1- \left(1-e^{-m  x}\right)^{m} \leq \mathbb{P}[H^{\sf L} > x] \leq  1-  \left(1-e^{-m \kappa x}\right)^{m},
        \end{align}
        where $\kappa=(m!)^{-\frac{1}{m}}$.
        We note that the equality holds when  $m=1$, i.e., the derived bounds will reduce to the exact analytical coverage expression.
        Applying  the binomial expansion:
        \begin{align}
            \label{eq:ccdf_H0_bound_binomial}
	       1-   \left(1-e^{-m \kappa x}\right)^{m} = \sum_{\ell=1}^{m} \binom{m}{\ell}(-1)^{\ell+1 }e^{-  \ell m \kappa x},
        \end{align}
        and plugging \eqref{eq:ccdf_H0_bound} and  \eqref{eq:ccdf_H0_bound_binomial} into \eqref{eq:Pcov_homo_recall}, we obtain an upper bound of the coverage probability for $\gamma >0$ dB as
        \begin{align}
            \nonumber
            &P^{\sf cov}_{{\sf SIR}} (\gamma;\lambda, R_{\sf S}) 
            \\\nonumber
            &\leq \int_{\bx_i\in\mathcal{A}}\Bigg( p_{\sf L}(r_i) \sum_{\ell=1}^{m} \binom{m}{\ell}(-1)^{\ell+1 }\mathbb{E}_{I_{r_i}} \left[ e^{-\frac{\ell m \kappa r_i^{\alpha_{\sf L}}\gamma I_{r_i}}{G(r_i)}}\middle |  r_i\right] 
            \\
            & \quad+\big(1-p_{\sf L}(r_i)\big)\mathbb{E}_{I_{r_i}} \left[e^{- \frac{r_i^{\alpha_{\sf N}}\gamma I_{r_i}}{G(r_i)}}  \middle|  r_i \right]  \Bigg) \lambda{\rm d}{\bx_i}
            \\\nonumber
            &=2\pi\lambda\frac{R_{\sf S}}{R_{\sf E}}\int_{R_{\rm min}}^{R_{\rm max}}  \Bigg( p_{\sf L}(r) \sum_{\ell=1}^{m} \binom{m}{\ell}(-1)^{\ell+1 }\mathcal{L}_{{  I}_{r}}\!\left(\frac{\ell m\kappa \gamma {r}^{\alpha_{\sf L}}}{G(r)} \right)  
            \\
            &\quad + \big(1-p_{\sf L}(r)\big)\mathcal{L}_{{  I}_{r} }\left(\frac{\gamma r^{\alpha_{\sf N}}}{G(r)}\right)\Bigg)  r {\rm d} r
            \\        \label{eq:Pcov_upper}
            & = P_{\sf SIR}^{\sf cov, b}(\gamma; \lambda, R_{\sf S}, \kappa).
        \end{align}
        The lower bound is directly obtained by setting $\kappa = 1$.
\qed

\section{Proof of Proposition \ref{prop:Pcov_approx_step}} \label{app:Pcov_approx_step}
The interference Laplace with \eqref{eq:P_LOS_step} and \eqref{eq:G_step} is given as
         \begin{align}
            \nonumber
            &\mathcal{L}_{I_{r}}(s) 
            = \exp\left(-2\pi \lambda \frac{R_{\sf S}}{R_{\sf E}}\int_{R_{\rm min}}^{R_{\rm max} } \left( 1-   p_{\sf L}(v)\frac{1}{ \left(1+\frac{sv^{-\alpha_{\sf L}}G_{\sf L}}{m}\right)^{m}} \right.\right.
            \\ \nonumber
            &\quad \left.\left.- (1-p_{\sf L}(v))\frac{1}{1+sv^{-\alpha_{\sf N}}G_{\sf N}}\right) v~ {\rm d}v \right)
            \\ \nonumber
            &= \exp\left(-2\pi \lambda \frac{R_{\sf S}}{R_{\sf E}}\int_{R_{\rm min}}^{R_{\rm los} } \left( 1-\frac{1}{ \left(1+\frac{sv^{-\alpha_{\sf L}}G_{\sf L}}{m}\right)^{m}}\right) v~ {\rm d}v \right.
            \\ \nonumber
            &\quad \left.-2\pi \lambda \frac{R_{\sf S}}{R_{\sf E}}\int_{R_{\rm los}}^{R_{\rm max} } \left( 1- \frac{1}{1+sv^{-\alpha_{\sf N}}G_{\sf N}}\right) v~ {\rm d}v \right)
            \\ \nonumber
            &\stackrel{(a)}= \exp\left(-\pi \lambda \frac{R_{\sf S}}{R_{\sf E}}\left(\frac{s G_{\sf L}}{m}\right)^{\frac{2}{\alpha_{\sf L}}}\int_{\left(\frac{s G_{\sf L}}{m}\right)^{-\frac{2}{\alpha_{\sf L}}}R^2_{\rm min}}^{\left(\frac{sG_{\sf L}}{m}\right)^{-\frac{2}{\alpha_{\sf L}}}R^2_{\rm los} }  1-\frac{1}{ \left(1+u^{-\frac{\alpha_{\sf L}}{2}}\right)^{m}} {\rm d}u  \right.
            \\ \label{eq:InterferenceLaplace_step}
            &\quad \left. -\pi \lambda \frac{R_{\sf S}}{R_{\sf E}}(sG_{\sf N})^{\frac{2}{\alpha_{\sf N}}}\int_{(sG_{\sf N})^{-\frac{2}{\alpha_{\sf N}}}R^2_{\rm los}}^{(sG_{\sf N})^{-\frac{2}{\alpha_{\sf N}}}R^2_{\rm max} }  1- \frac{1}{1+u^{-\frac{\alpha_{\sf N}}{2}}}  {\rm d}u \right),
        \end{align}
        where $(a)$ comes from change of variable.
        Then with some abuse of notation, we approximate the interference Laplace \eqref{eq:InterferenceLaplace_step} with $s = xr^\alpha$ as 
        \begin{align}
            \label{eq:L_xra}
            &\mathcal{L}_{I_{r}}(xr^\alpha)  
            \\ \nonumber 
            &=\exp\left(-\pi \lambda \frac{R_{\sf S}}{R_{\sf E}}r^{\frac{2\alpha}{\alpha_{\sf L}}}\left(\frac{xG_{\sf L}}{m}\right)^{\frac{2}{\alpha_{\sf L}}}  \right.
            \\ \nonumber
            & \left.  \times \int_{\left(\frac{xG_{\sf L}}{m}\right)^{-\frac{2}{\alpha_{\sf L}}}\frac{R^2_{\rm min}}{r^{\frac{2\alpha}{\alpha_{\sf L}}}}}^{\left(\frac{xG_{\sf L}}{m}\right)^{-\frac{2}{\alpha_{\sf L}}}\frac{R^2_{\rm los}}{r^{\frac{2\alpha}{\alpha_{\sf L}}}} }  1- \frac{1}{ \left(1+u^{-\frac{\alpha_{\sf L}}{2}}\right)^{m}} {\rm d}u \right.
            \\\label{eq:Laplace_rewrite} 
            & \left.-\pi \lambda \frac{R_{\sf S}}{R_{\sf E}}r^{\frac{2\alpha}{\alpha_{\sf N}}}(xG_{\sf N})^{\frac{2}{\alpha_{\sf N}}}\int_{(xG_{\sf N})^{-\frac{2}{\alpha_{\sf N}}}\frac{R^2_{\rm los}}{r^{\frac{2\alpha}{\alpha_{\sf N}}}}}^{(xG_{\sf N})^{-\frac{2}{\alpha_{\sf N}}}\frac{R^2_{\rm max}}{r^{\frac{2\alpha}{\alpha_{\sf N}}}} }\!  1\!-\! \frac{1}{1\!+\!u^{-\frac{\alpha_{\sf N}}{2}}}  {\rm d}u \right)  
            \\ \nonumber
            &\!\stackrel{(b)}\approx \! \exp\left(\!-\pi \lambda \frac{R_{\sf S}}{R_{\sf E}}r^{\frac{2\alpha}{\alpha_{\sf L}}}\left(\frac{xG_{\sf L}}{m}\right)^{\frac{2}{\alpha_{\sf L}}} \right.
            \\ \nonumber
            &\left. \times \int_{\left(\frac{xG_{\sf L}}{m}\right)^{-\frac{2}{\alpha_{\sf L}}}\frac{R^2_{\rm min}}{(\epsilon R_{\rm max})^{\frac{2\alpha}{\alpha_{\sf L}}}}}^{\left(\frac{xG_{\sf L}}{m}\right)^{-\frac{2}{\alpha_{\sf L}}}\!\frac{R^2_{\rm los}}{(R_{\rm min}/\epsilon)^{\frac{2\alpha}{\alpha_{\sf L}}}} } \!\! 1\!-\!\frac{1}{ \left(1\!+\!u^{-\frac{\alpha_{\sf L}}{2}}\right)^{m}} {\rm d}u \right.
            \\\nonumber 
            &\left. -\pi \lambda \frac{R_{\sf S}}{R_{\sf E}}r^{\frac{2\alpha}{\alpha_{\sf N}}}(xG_{\sf N})^{\frac{2}{\alpha_{\sf N}}}\int_{(xG_{\sf N})^{-\frac{2}{\alpha_{\sf N}}}\frac{R^2_{\rm los}}{(\epsilon R_{\rm max})^{\frac{2\alpha}{\alpha_{\sf N}}}}}^{(xG_{\sf N})^{-\frac{2}{\alpha_{\sf N}}}\frac{R^2_{\rm max}}{(R_{\rm min}/\epsilon)^{\frac{2\alpha}{\alpha_{\sf N}}}} }  1\!-\! \frac{1}{1\!+\!u^{-\frac{\alpha_{\sf N}}{2}}}  {\rm d}u \!\right)
            \\\nonumber 
            & = \exp\left(-\pi \lambda \frac{R_{\sf S}}{R_{\sf E}}\left(r^{\frac{2\alpha}{\alpha_{\sf L}}}\rho_{\sf L}(x,\alpha;\epsilon) +r^{\frac{2\alpha}{\alpha_{\sf N}}}\rho_{\sf N}(x,\alpha;\epsilon) \right)\right) 
            \\ \label{eq:Laplace_approx}
            &=\mathcal{L}^{\sf a}_{I_{r}}\left(xr^\alpha\right),  
        \end{align}
            where $(b)$ follows from introducing $0<\epsilon \leq 1$.
            Since the range of $\epsilon$ should be valid for integral range, it should be
            \[\max\left(\frac{R_{\rm min}^{\frac{\alpha + \alpha_{\sf L}}{2\alpha}}}{R_{\rm max}^{\frac{1}{2}}R_{\rm los}^{\frac{\alpha_{\sf L}}{2\alpha}}}, \frac{R_{\rm min}^{\frac{1}{2}}R_{\rm los}^{\frac{\alpha_{\sf N}}{2\alpha}}}{R_{\rm max}^{\frac{\alpha + \alpha_{\sf N}}{2\alpha}}}\right)  \leq \epsilon \leq 1\] for $\alpha \in \{\alpha_{\sf L}, \alpha_{\sf N}\}$.
            
            Assuming $\alpha_{\sf N}/\alpha_{\sf L} = 2$, the approximated coverage probability $P_{\sf SIR}^{\sf cov,b}(\gamma; \lambda, R_{\sf S}, \kappa)$ in Proposition~\ref{prop:Pcov_approx} is further approximated with $\mathcal{L}^{\sf a}_{I_r}(xr^\alpha)$ in \eqref{eq:Laplace_approx} and solved as  in \eqref{eq:Pcov_approx_proof} which is shown at the top of the next page,
 \begin{figure*}
    \begin{align}
        \nonumber
        P^{\sf cov,b}_{{\sf SIR}} (\gamma;\lambda, R_{\sf S}, \kappa)  
       &\approx 2\pi\lambda\frac{R_{\sf S}}{R_{\sf E}} \sum_{\ell=1}^{m} \binom{m}{\ell}(-1)^{\ell+1 }\int_{R_{\rm min}}^{R_{\rm los}}  \mathcal{L}^{\sf a}_{{  I}_{r}}\!\left(\frac{\ell m \gamma \kappa}{G_{\sf L}}{r}^{\alpha_{\sf L}} \right) r {\rm d} r  
        + 2\pi\lambda\frac{R_{\sf S}}{R_{\sf E}}\int_{R_{\rm los}}^{R_{\rm max}} \mathcal{L}^{\sf a}_{{  I}_{r} }\left(\frac{\gamma}{G_{\sf N}}r^{\alpha_{\sf N}}\right) r {\rm d} r, 
        \\ \nonumber
        & =  \sum_{\ell=1}^{m} \binom{m}{\ell}(-1)^{\ell+1 }\left[\frac{1}{\rho_{\sf L}(z, \alpha_{\sf L})}\left\{e^{-\pi \lambda \frac{R_S}{R_E}R_{\rm min}\left(\rho_{\sf N}(z, \alpha_{\sf L})+R_{\rm min}\rho_{\sf L}(z, \alpha_{\sf L})\right)}\!-\!e^{-\pi \lambda \frac{R_S}{R_E}R_{\rm los}\left(\rho_{\sf N}(z, \alpha_{\sf L})+R_{\rm los}\rho_{\sf L}(z, \alpha_{\sf L})\right)}\right\} \right.
        \\\nonumber
        &\left. \quad + \frac{\pi \rho_{\sf N}(z, \alpha_{\sf L})}{2\rho_{\sf L}(z, \alpha_{\sf L})}\sqrt{\frac{\lambda}{\rho_{\sf L}(z, \alpha_{\sf L})}\frac{R_{\sf S}}{R_{\sf E}}} e^{\pi \lambda\frac{R_{\sf S}}{4R_{\sf E}} \frac{\left(\rho_{\sf N}(z, \alpha_{\sf L})\right)^2}{\rho_{\sf L}(z, \alpha_{\sf L})}}\left\{{\rm erf}\left(\sqrt{\frac{\pi \lambda R_{\sf S}}{4R_{\sf E}}}\frac{\rho_{\sf N}(z, \alpha_{\sf L}) + 2R_{\rm min} \rho_{\sf L}(z, \alpha_{\sf L})}{\sqrt{\rho_{\sf L}(z, \alpha_{\sf L})}}\right) \right.\right.
        \\\nonumber
        &\left.\left. \quad -{\rm erf}\left(\sqrt{\frac{\pi \lambda R_{\sf S}}{4R_{\sf E}}}\frac{\rho_{\sf N}(z, \alpha_{\sf L}) + 2R_{\rm los} \rho_{\sf L}(z, \alpha_{\sf L})}{\sqrt{\rho_{\sf L}(z, \alpha_{\sf L})}}\right) \right\}\Bigg|_{z=\frac{\ell m \gamma \kappa}{G_{\sf L}}}\right] 
        + \frac{\pi}{2}\sqrt{\frac{R_{\sf S}\lambda }{R_{\sf E}\rho_{\sf L}(\frac{\gamma}{G_{\sf N}} \alpha_{\sf N})}}e^{\frac{\pi\lambda R_{\sf S}\left(\rho_{\sf N}(\frac{\gamma}{G_{\sf N}}, \alpha_{\sf N})\right)^2}{4R_{\sf E}\rho_{\sf L}(\frac{\gamma}{G_{\sf N}}, \alpha_{\sf N})}} 
        \\\nonumber
        & \quad \times \left\{{\rm erf}\left(\sqrt{\frac{\pi \lambda R_{\sf S}}{4R_{\sf E}}}\frac{\rho_{\sf N}(\frac{\gamma}{G_{\sf N}}, \alpha_{\sf N}) + 2R_{\rm max}^2 \rho_{\sf L}(\frac{\gamma}{G_{\sf N}}, \alpha_{\sf N})}{\sqrt{\rho_{\sf L}(\frac{\gamma}{G_{\sf N}}, \alpha_{\sf N})}}\right)  -{\rm erf}\left(\sqrt{\frac{\pi \lambda R_{\sf S}}{4R_{\sf E}}}\frac{\rho_{\sf N}( \frac{\gamma}{G_{\sf N}}, \alpha_{\sf N}) + 2R^2_{\sf L} \rho_{\sf L}(\frac{\gamma}{G_{\sf N}}, \alpha_{\sf N})}{\sqrt{\rho_{\sf L}( \frac{\gamma}{G_{\sf N}}, \alpha_{\sf N})}}\right)\right\}       
        \\ \nonumber 
        &\stackrel{(c)}\approx \frac{1}{4}\sqrt{\frac{c\pi}{\rho_{\sf L}(\frac{\gamma}{G_{\sf N}}, \alpha_{\sf N};\epsilon)}}\left(\frac{1}{3}e^{-c\Psi_{\sf 1}(R_{\rm los}^2,\frac{\gamma}{G_{\sf N}},\alpha_{\sf N})} -\frac{1}{3}e^{-c\Psi_{\sf 1}(R_{\rm max}^2,\frac{\gamma}{G_{\sf N}},\alpha_{\sf N})} +e^{-c\Psi_{\sf 2}(R_{\rm los}^2,\frac{\gamma}{G_{\sf N}},\alpha_{\sf N})} -e^{-c\Psi_{\sf 2}(R_{\rm max}^2,\frac{\gamma}{G_{\sf N}},\alpha_{\sf N})}\right) 
        \\\nonumber
        &\quad +\sum_{\ell=1}^{m}\! \binom{m}{\ell}\!{(-1)^{\ell+1 }}\!\Bigg[ \!\frac{\sqrt{c\pi} \rho_{\sf N}(z, \alpha_{\sf L};\epsilon)}{4\rho_{\sf L}(z, \alpha_{\sf L};\epsilon)^{\frac{3}{2}}}\!\left(\! \frac{1}{3}e^{-c\Psi_{\sf 1}\!(R_{\rm los},z,\alpha_{\sf L})}\! -\!\frac{1}{3}e^{-c\Psi_{\sf 1}\!(R_{\rm min},z,\alpha_{\sf L})} \!+ \!e^{-c\Psi_{\sf 2}\!(R_{\rm los},z,\alpha_{\sf L})} \!-\! e^{-c\Psi_{\sf 2}\!(R_{\rm min},z,\alpha_{\sf L})}\!\right) 
        \\ \label{eq:Pcov_approx_proof}
        &\quad + \frac{e^{-c \Psi_{\sf 1}\!(R_{\rm min},z,\alpha_{\sf L})}\!-\!e^{-c\Psi_{\sf 1}\!(R_{\rm los},z,\alpha_{\sf L})}}{\rho_{\sf L}(z, \alpha_{\sf L};\epsilon)}\bigg|_{z=\frac{\ell m \gamma \kappa}{G_{\sf L}}}\!\Bigg]
    \end{align}
           \noindent\rule{\textwidth}{0.5pt}
    \end{figure*}
    where $(c)$ comes from 
    \[1-{\rm erfc}(x) \approx \frac{1}{6}e^{-x^2} + \frac{1}{2}e^{-\frac{4}{3}x^2}, \quad x >0. \] This completes the proof. 
\qed

\section{Proof of Corollary \ref{cor:Pcov_lb_step}} \label{app:Pcov_lb_step}
Setting $\kappa = \epsilon = 1$, we directly obtain \eqref{eq:Pcov_lb_step} from \eqref{eq:Pcov_approx_step}:
when $\kappa =1$, we have the lower bound derived in Theorem~\ref{thm:SandwichBounds}.
When $\epsilon =1$, the integrals in $\rho_{\sf L}(x,\alpha;\epsilon)$ and $\rho_{\sf N}(x,\alpha;\epsilon)$ are maximized in terms of the integral range so that the interference Laplace $\mathcal{L}_{I_r}(xr^\alpha)$ in \eqref{eq:Laplace_rewrite} is minimized, which further leads to the lower bound of the coverage probability.
 \qed

\section{Proof of Theorem \ref{thm:lowerbound_homo}} \label{app:lowerbound_homo}
    Assuming homogeneous and Rayleigh fading channels, we consider $R_{\rm los} = R_{\sf max}$  with $m=1$ in this proof.
    Without loss of generality, we denote $\alpha_{\sf L}$ as $\alpha$ since all channels have the same pathloss exponent.
    In the considered case, the interference Laplace in \eqref{eq:Laplace_rewrite} reduces to 
        \begin{align}
            \nonumber
            &\mathcal{L}^{\sf hm}_{I_{r}}\left(xr^{\alpha}\right) \Big|_{x = \gamma/G_{\rm L}} 
            \\\nonumber
            &=\exp\left(-\pi \lambda \frac{R_{\sf S}}{R_{\sf E}}r^2{\gamma}^{\frac{2}{\alpha}}\int_{\gamma^{-\frac{2}{\alpha}}\frac{R^2_{\rm min}}{r^2}}^{\gamma^{-\frac{2}{\alpha}}\frac{R^2_{\rm max}}{r^2} }  1-\frac{1}{1+u^{-\frac{\alpha}{2}}} {\rm d}u \right)
            \\ \label{eq:homo_lower}
            &\geq  \exp\left(-\pi \lambda \frac{R_{\sf S}}{R_{\sf E}}r^2\rho^{\sf hm}(\gamma)\right).
        \end{align}
        Then using the lower bound in Theorem~\ref{thm:SandwichBounds}, the coverage probability  is lower bounded by
        \begin{align}
            P_{\sf SIR}^{\sf cov}(\gamma;\lambda, R_{\sf S})&\geq P^{\sf cov, b}_{{\sf SIR}} (\gamma;\lambda, R_{\sf S}, 1)  
            \\
            &\stackrel{(a)}\geq 2\pi\lambda\frac{R_{\sf S}}{R_{\sf E}} \int_{R_{\rm min}}^{R_{\rm max}}e^{-\pi \lambda \frac{R_{\sf S}}{R_{\sf E}}r^2\rho^{\sf hm}\left({\gamma}\right)}  r {\rm d} r
            \\\nonumber
            & = \!\frac{1}{\rho^{\sf hm}({\gamma})}\!\left(e^{-\frac{\pi\lambda  R_{\sf S}}{R_{\sf E}}\rho^{\sf hm}({\gamma})R_{\rm min}^2}\! -\! e^{-\frac{\pi\lambda R_{\sf S}}{R_{\sf E}}\rho^{\sf hm}({\gamma})R_{\rm max}^2}\!\right),
        \end{align}
        where $(a)$ follows from replacing the interference Laplace $\mathcal{L}_{I_r}$ in \eqref{eq:Pcov_sandwich} with \eqref{eq:homo_lower}.
\qed

\section{Proof of Theorem \ref{thm:OptimalDensity}} \label{app:OptimalDensity}
        The lower bound \eqref{eq:lowerbound_homo} in Theorem~\ref{thm:lowerbound_homo} is a unimodal function with respect to 
respect to $\lambda > 0$ which increases for $\lambda \in [0,\lambda^\star]$ and decreases for  $\lambda \in [\lambda^\star, \infty)$ \cite{park2022tractable} where $\lambda^\star$ is the stationary point of \eqref{eq:lowerbound_homo}:
\begin{align}
    \lambda^\star &= \frac{2R_{\sf E}\log \left(\frac{R_{\rm max}}{R_{\rm min}}\right)}{\pi R_{\sf S} \rho^{\sf hm}({\gamma})(R_{\rm max}^2 - R_{\rm min}^2)}.
\end{align}
Using $R_{\rm max} = \sqrt{R_{\sf S}^2 - R_{\sf E}^2}$, $R_{\rm min} = h$, and $R_{\sf S} = R_{\sf E} + h$, we obtain \eqref{eq:OptimalDensity}.
\qed

\bibliographystyle{IEEEtran}
\bibliography{LEO_Shadowing.bib}

\end{document}